\documentclass{article}

\usepackage[preprint]{nips_2018}

\usepackage[utf8]{inputenc} 
\usepackage[T1]{fontenc}    
\usepackage[hidelinks]{hyperref}
\usepackage{url}            
\usepackage{booktabs}       
\usepackage{amsfonts}       
\usepackage{nicefrac}       
\usepackage{microtype}      

\usepackage{amssymb,amsmath,amsthm}
\usepackage{graphicx}
\newtheorem{theorem}{Theorem}
\newtheorem{definition}[theorem]{Definition}
\newtheorem{corollary}[theorem]{Corollary}
\newtheorem{lemma}[theorem]{Lemma}

\newcommand\onenorm[1]{\|#1\|_1}
\newcommand\twonorm[1]{\|#1\|_2}
\newcommand\Fnorm[1]{\|#1\|_F}

\newcommand\g[1]{g\!\left(#1\right)}

\newcommand\eps{\varepsilon}
\newcommand\REALS{\mathbb{R}}
\newcommand\rng[1]{\mathrm{range}_{#1}}
\newcommand\ranges{\mathrm{ranges}}
\newcommand\fcal{\mathcal{F}}
\newcommand\wcal{\mathcal{W}}

\newcommand\nnz{\operatorname{\mathtt{nnz}}}
\newcommand\poly{\operatorname{poly}}

\title{On Coresets for Logistic Regression}

\author{
  Alexander~Munteanu\\
  Department of Computer Science\\
  TU Dortmund University\\
  44227 Dortmund, Germany \\
  \texttt{alexander.munteanu@tu-dortmund.de} \\
  \And
  Chris~Schwiegelshohn \\
  Department of Computer Science\\
  Sapienza University of Rome\\
  00185 Rome, Italy \\
  \texttt{schwiegelshohn@diag.uniroma1.it} \\
  \And
  Christian~Sohler \\
  Department of Computer Science\\
  TU Dortmund University\\
  44227 Dortmund, Germany \\
  \texttt{christian.sohler@tu-dortmund.de} \\
  \And
  David P.~Woodruff \\
  Department of Computer Science\\
  Carnegie Mellon University\\
  Pittsburgh, PA 15213, USA \\
  \texttt{dwoodruf@cs.cmu.edu} \\
}

\begin{document}

\maketitle

\begin{abstract}
  Coresets are one of the central methods to facilitate the analysis of large data. We continue a recent line of research applying the theory of coresets to logistic regression. 
  First, we show the negative result that no 
  strongly sublinear sized coresets exist for logistic regression. 
  To deal with intractable worst-case instances 
  we introduce a complexity measure $\mu(X)$, which quantifies the hardness of compressing a data set for logistic regression. 
  $\mu(X)$ has an intuitive statistical interpretation that may be of independent interest.
  For data sets with bounded $\mu(X)$-complexity, we show that a novel sensitivity sampling scheme produces the first provably sublinear $(1\pm\eps)$-coreset.
  We illustrate the performance of our method by comparing to uniform sampling as well as to state of the art methods in the area. The experiments are conducted on real world benchmark data for logistic regression.
\end{abstract}
\allowdisplaybreaks
\section{Introduction}

Scalability is one of the central challenges of modern data analysis and machine learning.
Algorithms with polynomial running time might be regarded as efficient in a conventional sense, but nevertheless become intractable when facing massive data sets.
As a result, performing data reduction techniques in a preprocessing step to speed up a subsequent optimization problem has received considerable attention.
A natural approach is to sub-sample the data according to a certain probability distribution. 
This approach has been successfully applied to a variety of problems including clustering~\citep{LangbergS10,FeL11,BargerF16,BachemLK18}, mixture models \citep{FeldmanFK11,LucicBK16}, low rank approximation~\citep{CMM17}, spectral approximation~\citep{AlM15,LiMP13}, and Nystr\"om methods~\citep{AlM15,MM17}.

The unifying feature of these works is that the probability distribution is based on the sensitivity score of each point. Informally, the sensitivity of a point corresponds to the importance of the point with respect to the objective function we wish to minimize. If the total sensitivity, i.e., the sum of all sensitivity scores $\mathfrak{S}$, is bounded by a reasonably small value ${S}$, there exists a collection of input points known as a coreset with very strong aggregation properties. Given any candidate solution (e.g., a set of $k$ centers for $k$-means, or a hyperplane for linear regression), the objective function computed on the coreset evaluates to the objective function of the original data up to a small multiplicative error. See Sections~\ref{sec:prelim} and~\ref{sec:sensitivity} for formal definitions of sensitivity and coresets.

\textbf{Our Contribution}
We investigate coresets for logistic regression within the sensitivity framework. 
Logistic regression is an instance of a generalized linear model where we are given data $Z\in\REALS^{n\times d}$, and labels $Y\in\{-1,1\}^n$. The optimization task consists of minimizing the negative log-likelihood $\sum\nolimits_{i=1}^{n}\ln(1+\exp(-Y_i Z_i\beta))$ with respect to the parameter $\beta\in\REALS^d$ \citep{McCullaghN89}.

$\bullet$ Our first contribution is an impossibility result: logistic regression has no sublinear streaming algorithm. Due to a standard reduction between coresets and streaming algorithms, this also implies that logistic regression admits no coresets or bounded sensitivity scores in general.

$\bullet$ Our second contribution is an investigation of available sensitivity sampling distributions for logistic regression. For points with large contribution, where $-Y_iZ_i\beta \gg 0$, the objective function increases by a term almost linear in $-Y_iZ_i\beta$. This questions the use of sensitivity scores designed for problems with squared cost functions such as $\ell_2$-regression, $k$-means, and $\ell_2$-based low-rank approximation. Instead, we propose sampling from a mixture distribution with one component proportional to the \emph{square root} of the $\ell^2_2$ leverage scores. Though seemingly similar to the sampling distributions of e.g. \citep{BargerF16,BachemLK18} at first glance, it is important to note that sampling according to $\ell_2^2$ scores is different from sampling according to their square roots. The former is good for $\ell_2$-related loss functions, while the latter preserves $\ell_1$-related functions such as the linear part of the original logistic regression loss function studied here. The other mixture component is uniform sampling to deal with the remaining domain, where the cost function consists of an exponential decay towards zero. Our experiments show that this distribution outperforms uniform and $k$-means based sensitivity sampling by a wide margin on real data sets. The algorithm is space efficient, and can be implemented in a variety of models used to handle large data sets such as $2$-pass streaming, and massively parallel frameworks such as Hadoop and MapReduce, and can be implemented in input sparsity time, $\tilde O(\nnz(Z))$, the number of non-zero entries of the data \citep{ClarksonW13}.

$\bullet$ Our third contribution is an analysis of our sampling distribution for a parametrized class of instances we call $\mu$-complex, placing our work in the framework of \emph{beyond worst-case analysis}~\citep{BalcanMRR15,Roughgarden17}. The parameter $\mu$ roughly corresponds to the ratio between the log of correctly estimated odds and the log of incorrectly estimated odds. The condition of small $\mu$ is justified by the fact that for instances with large $\mu$, logistic regression exhibits methodological problems like imbalance and separability, cf. \citep{MethaP95,HeinzeS02}. We show that the total sensitivity of logistic regression can be bounded in terms of $\mu$, and that our sampling scheme produces the first coreset of provably sublinear size, provided that $\mu$ is small.

\textbf{Related Work}
There is more than a decade of extensive work on sampling based methods relying on the sensitivity framework for $\ell_2$-regression \citep{DrineasMM06,DrineasMM08,LiMP13,CLMMPS15} and $\ell_1$-regression \citep{Clarkson05,SohlerW11,ClarksonDMMMW16}. These were generalized to $\ell_p$-regression for all $p\in [1,\infty)$ \citep{DasguptaDHKM09,WoodruffZ13}. More recent works study sampling methods for $M$-estimators \citep{ClarksonW15,ClarksonW15_focs} and extensions to generalized linear models \citep{HugginsCB16,MolinaMK17}. The contemporary theory behind coresets has been applied to logistic regression, first by \cite{ReddiPS15} using first order gradient methods, and subsequently via sensitivity sampling by \cite{HugginsCB16}.
In the latter work, the authors recovered the result that bounded sensitivity scores for logistic regression imply coresets. Explicit sublinear bounds on the sensitivity scores, as well as an algorithm for computing them, were left as an open question. Instead, they proposed using sensitivity scores derived from any $k$-means clustering for logistic regression. While high sensitivity scores of an input point for $k$-means provably do not imply a high sensitivity score of the same point for logistic regression, the authors observed that they can outperform uniform random sampling on a number of instances with a clustering structure. Recently and independently of our work, \cite{FeldmanT18} gave a coreset construction for logistic regression in a more general framework. Our construction is without regularization and therefore can be also applied for any regularized version of logistic regression, but we have constraints regarding the $\mu$-complexity of the input. Their result is for $\ell^2_2$-regularization, which significantly changes the objective and does not carry over to the unconstrained version. They do not constrain the input but the domain of optimization is bounded. This indicates that both results differ in many important points and are of independent interest.

All proofs and additional plots from the experiments are in the appendices \ref{app:proofs} and \ref{app:experiments}, respectively.

\section{Preliminaries and Problem Setting}
\label{sec:prelim}
In logistic regression we are given a data matrix $Z\in\REALS^{n\times d}$, and labels $Y\in\{-1,1\}^n$. 
Logistic regression has a negative log-likelihood \citep{McCullaghN89}
\begin{align*}
\mathcal{L}(\beta | Z,Y) = \sum\nolimits_{i=1}^{n}\ln(1+\exp(-Y_i Z_i\beta))
\end{align*}
which from a learning and optimization perspective, is the objective function that we would like to minimize over $\beta\in \REALS^d$. 
For brevity we fold for all $i\in[n]$ the labels $Y_i$ as well as the factor $-1$ in the exponent into $X\in\REALS^{n\times d}$ comprising row vectors $x_i=-Y_iZ_i$.
Let $g(z)=\ln(1+\exp(z))$. For technical reasons we deal with a weighted version for weights $w\in\REALS_{> 0}^n$, where each weight satisfies $w_i > 0$. Any positive scaling of the all ones vector $\mathbf{1}=\{1\}^n$ corresponds to the unweighted case. We denote by $D_w$ a diagonal matrix carrying the entries of $w$, i.e., $(D_w)_{ii}=w_i$, so that multiplying $D_w$ to a vector or matrix has the effect of scaling row $i$ by a factor of $w_i$. The objective function becomes
\begin{align*}
f_w(X\beta) &= \sum\nolimits_{i=1}^{n} w_i g(x_i\beta) = \sum\nolimits_{i=1}^{n} w_i \ln(1+\exp(x_i\beta)).
\end{align*}
In this paper we assume we have a very large number of observations in a moderate number of dimensions, that is, $n\gg d$. In order to speed up the computation and to lower memory and storage requirements we would like to significantly reduce the number of observations without losing much information in the original data. A suitable data compression reduces the size to a sublinear number of $o(n)$ data points while the dependence on $d$ and the approximation parameters may be polynomials of low degree. To achieve this, we design a so-called coreset construction for the objective function. A coreset is a possibly (re)weighted and significantly smaller subset of the data that approximates the objective value for any possible query points. More formally, we define coresets for the weighted logistic regression function.
\begin{definition}[$(1\pm\eps)$-coreset for logistic regression]
	\label{def:coreset}
	Let $X\in\REALS^{n\times d}$ be a set of points weighted by $w\in \REALS_{> 0}^n$.
	Then a set $C\in\REALS^{k\times d}$, (re)weighted by $u\in \REALS_{> 0}^k$, is a $(1\pm\eps)$-coreset of $X$ for $f_w$, if $k\ll n$ and
	\[\forall \beta\in \REALS^d: |f_w(X\beta)-f_u(C\beta)|\leq \eps\cdot f_w(X\beta).\]
\end{definition}
\textbf{$\mu$-Complex Data Sets}
\label{sec:assumption}
We will see in Section \ref{sec:LBs} that in general, there is no sublinear one-pass streaming algorithm approximating the objective function up to any finite constant factor. More specifically there exists no sublinear summary or coreset construction that works for all data sets. For the sake of developing coreset constructions that work \emph{reasonably well}, as well as conducting a formal analysis beyond worst-case instances, we introduce a measure $\mu$ that quantifies the \emph{complexity} of compressing a given data set.
\begin{definition}
	\label{assumption}
	Given a data set $X\in\REALS^{n\times d}$ weighted by $w\in \REALS_{> 0}^n$ and a vector $\beta\in \REALS^d$ let $(D_wX\beta)^-$ denote the vector comprising only the negative entries of $D_wX\beta$. Similarly let $(D_wX\beta)^+$ denote the vector of positive entries. We define for $X$ weighted by $w$
	$$\mu_w(X) = \sup\limits_{\beta\in \REALS^d\setminus\{0\}} \frac{\onenorm{(D_wX\beta)^+}}{\onenorm{(D_wX\beta)^-}}.$$
	$X$ weighted by $w$ is called $\mu$-complex if $\mu_w(X)\leq \mu$.
\end{definition}
The size of our $(1\pm\eps)$-coreset constructions for logistic regression for a given $\mu$-complex data set $X$ will have low polynomial dependency on $\mu,d,1/\eps$ but only sublinear dependency on its original size parameter $n$. So for $\mu$-complex data sets having small $\mu(X)\leq \mu$ we have the first $(1\pm\eps)$-coreset of provably sublinear size.
The above definition implies, for $\mu(X)\leq\mu$, the following inequalities. The reader should keep in mind that for all $\beta\in\REALS^d$
\begin{align*}
\mu^{-1}\onenorm{(D_wX\beta)^-}\leq\onenorm{(D_wX\beta)^+} \leq \mu \onenorm{(D_wX\beta)^-}\,.
\end{align*}
We conjecture that computing the value of $\mu(X)$ is hard. However, it can be approximated in polynomial time. It is not necessary to do so in practical applications, but we include this result for those who wish to evaluate whether their data has nice $\mu$-complexity.
\begin{theorem}
	\label{thm:muLP}
	Let $X\in\REALS^{n\times d}$ be weighted by $w\in\REALS^n_{>0}$.
	Then a $\poly(d)$-approximation to the value of $\mu_w(X)$ can be computed in $O(\poly(nd))$ time.
\end{theorem}
The parameter $\mu(X)$ has an intuitive interpretation and might be of independent interest. 
The odds of a binary random variable $V$ are defined as
$\frac{\mathbb{P}[V=1]}{\mathbb{P}[V=0]}.$
The model assumption of logistic regression is that for every sample $X_i$, the logarithm of the odds is a linear function of $X_i\beta$. For a candidate $\beta$, multiplying all odds and taking the logarithm is then exactly $\|X\beta\|_1$. Our definition now relates the probability mass due to the incorrectly predicted odds and the probability mass due to the correctly predicted odds. We say that the ratio between these two is upper bounded by $\mu$. For logistic regression, assuming they are within some order of magnitude is not uncommon. One extreme is the (degenerate) case where the data set is exactly separable. Choosing $\beta$ to parameterize a separating hyperplane for which $X\beta$ is all positive, implies that $\mu(X)=\infty$. Another case is when we have a large ratio between the number of positively and negatively labeled points which is a lower bound to $\mu$. Under either of these conditions, logistic regression exhibits methodological weaknesses due to the separation or imbalance between the given classes, cf. \citep{MethaP95,HeinzeS02}.

\section{Lower Bounds}
\label{sec:LBs}
At first glance, one might think of taking a uniform sample as a coreset. We demonstrate and discuss on worst-case instances in Appendix \ref{dis:uniform} that this won't work in theory or in practice. In the following we will show a much stronger result, namely that no efficient streaming algorithms or coresets for logistic regression can exist in general, even if we assume that the points lie in $2$-dimensional Euclidean space. To this end we will reduce from the I{\footnotesize NDEX} communication game. In its basic variant, there exist two players Alice and Bob. Alice is given a binary bit string $x\in\{0,1\}^n$ and Bob is given an index $i\in [n]$. The goal is to determine the value of $x_i$ with constant probability while using as little communication as possible. Clearly, the difficulty of the problem is inherently one-way; otherwise Bob could simply send his index to Alice. If the entire communication consists of only a single message sent by Alice to Bob, the message must contain $\Omega(n)$ bits~\citep{KNR99}.
\begin{theorem}\label{LB:streaming}
	Let $Z\in \mathbb{R}^{n\times 2},Y\in\{-1,1\}^n$ be an instance of logistic regression in $2$-dimensional Euclidean space.
	Any one-pass streaming algorithm that approximates the optimal solution of logistic regression up to any finite multiplicative approximation factor requires $\Omega(n/\log n)$ bits of space.
\end{theorem}
A similar reduction also holds if Alice's message consists of points forming a coreset. Hence, the following corollary holds.
\begin{corollary}\label{LB:coreset}
	Let $Z\in \mathbb{R}^{n\times 2},Y\in\{-1,1\}^n$ be an instance of logistic regression in $2$-dimensional Euclidean space. Any coreset of $Z,Y$ for logistic regression consists of at least $\Omega(n/\log n)$ points.
\end{corollary}
We note that the proof can be slightly modified to rule out any finite additive error as well. This indicates that the notion of \emph{lightweight} coresets with multiplicative and additive error \citep{BachemLK18} is not a sufficient relaxation.
Independently of our work \cite{FeldmanT18} gave a linear lower bound in a more general context based on a worst case instance to the sensitivity approach due to \cite{HugginsCB16}. Our lower bounds and theirs are incomparable; they show that if a coreset can only consist of input points it comprises the entire data set in the worst-case. We show that no coreset with $o(n/\log n)$ can exist, irrespective of whether input points are used. While the distinction may seem minor, a number of coreset constructions in literature necessitate the use of non-input points, see~\citep{AgarwalHV04} and~\citep{FeldmanSS13}.

\section{Sampling via Sensitivity Scores}
\label{sec:sensitivity}
Our sampling based coreset constructions are obtained with the following approach, called sensitivity sampling. Suppose we are given a data set $X\in\REALS^{n\times d}$ together with weights $w\in \REALS_{> 0}^n$ as in Definition \ref{def:coreset}. Recall the function under study is $f_w(X\beta)=\sum\nolimits_{i=1}^n w_i\cdot g(x_i\beta)$. Associate with each point $x_i$ the function $g_i(\beta)=g(x_i\beta)$. Then we have the following definition.
\begin{definition}{\citep{LangbergS10}}
	\label{def:sensitivities}
	Consider a family of functions $\fcal=\{g_1,\ldots,g_n\}$ mapping from $\REALS^d$ to $[0,\infty)$ and weighted by $w\in\REALS_{> 0}^n$. The sensitivity of $g_i$ for $f_w(\beta)=\sum\nolimits_{i=1}^{n} w_i g_i(\beta)$ is
	\begin{align}
	\label{eqn:sensitivities}
	\varsigma_i = \sup \frac{w_i g_i(\beta)}{f_w(\beta)}
	\end{align}
	where the $\sup$ is over all $\beta\in\REALS^d$ with $f_w(\beta) > 0$. If this set is empty then $\varsigma_i=0$. The total sensitivity is $\mathfrak{S} = \sum\nolimits_{i=1}^{n} \varsigma_i$.
\end{definition}
The sensitivity of a point measures its worst-case importance for approximating the objective function on the entire input data set. Performing importance sampling proportional to the sensitivities of the input points thus yields a good approximation. Computing the sensitivities is often intractable and involves solving the original optimization problem to near-optimality, which is the problem we want to solve in the first place, as pointed out in \citep{BravermanFL16}. To get around this, it was shown that any upper bound on the sensitivities $s_i\geq\varsigma_i$ also has provable guarantees. However, the number of samples needed depends on the total sensitivity, that is, the sum of their estimates $S=\sum\nolimits_{i=1}^{n} s_i \geq \sum\nolimits_{i=1}^{n}\varsigma_i = \mathfrak{S}$, so we need to carefully control this quantity. Another complexity measure that plays a crucial role in the sampling complexity is the VC dimension of the range space induced by the set of functions under study.
\begin{definition}
	A range space is a pair $\mathfrak{R}=(\fcal,\ranges)$ where $\fcal$ is a set and $\ranges$ is a family of subsets of $\fcal$. The VC dimension $\Delta(\mathfrak{R})$ of $\mathfrak{R}$ is the size $|G|$ of the largest subset $G\subseteq \fcal$ such that $G$ is shattered by $\ranges$, i.e., $\left| \{G\cap R\mid R\in \ranges \} \right| = 2^{|G|}.$
\end{definition}
\begin{definition}
	Let $\fcal$ be a finite set of functions mapping from $\REALS^d$ to $\REALS_{\geq 0}$. For every $\beta\in\REALS^d$ and $r\in \REALS_{\geq 0}$, let $\rng{\fcal}(\beta,r) = \{ f\in \fcal\mid f(\beta)\geq r\}$, and $\ranges(\fcal)=\{\rng{\fcal}(\beta,r)\mid \beta\in\REALS^d, r\in \REALS_{\geq 0} \}$, and $\mathfrak{R}_{\fcal}=(\fcal,\ranges(\fcal))$ be the range space induced by $\fcal$.
\end{definition}
Recently a framework combining the sensitivity scores with a theory on the VC dimension of range spaces was developed in \citep{BravermanFL16}. For technical reasons we use a slightly modified version.
\begin{theorem}
	\label{thm:sensitivity}
	Consider a family of functions $\fcal=\{f_1,\ldots,f_n\}$ mapping from $\REALS^d$ to $[0,\infty)$ and a vector of weights $w\in\REALS_{> 0}^n$. Let $\eps,\delta\in(0,1/2)$. Let $s_i\geq \varsigma_i$. Let $S=\sum\nolimits_{i=1}^{n} s_i \geq \mathfrak{S}$. Given $s_i$ one can compute in time $O(|\fcal|)$ a set $R\subset \fcal$ of $$O\left( \frac{S}{\eps^2}\left( \Delta \log S + \log \left(\frac{1}{\delta}\right) \right) \right)$$ weighted functions such that with probability $1-\delta$ we have for all $\beta\in \REALS^d$ simultaneously $$\left| \sum_{f\in \fcal} w_i f_i(\beta) - \sum_{f\in R} u_i f_i(\beta) \right| \leq \eps \sum_{f\in \fcal} w_i f_i(\beta).$$
	where each element of $R$ is sampled i.i.d. with probability $p_j=\frac{s_j}{S}$ from $\fcal$, $u_i = \frac{Sw_j}{s_j|R|}$ denotes the weight of a function $f_i\in R$ that corresponds to $f_j\in\fcal$, and where $\Delta$ is an upper bound on the VC dimension of the range space $\mathfrak{R}_{\fcal^*}$ induced by $\fcal^*$ that can be obtained by defining $\fcal^*$ to be the set of functions $f_j\in\fcal$ where each function is scaled by $\frac{Sw_j}{s_j|R|}$.
\end{theorem}
Now we show that the VC dimension of the range space induced by the set of functions studied in logistic regression can be related to the VC dimension of the set of linear classifiers. We first start with a fixed common weight and generalize the result to a more general finite set of distinct weights.
\begin{lemma}
	\label{lem:vc1}
	Let $X\in\REALS^{n\times d}, c\in \REALS_{>0}$.
	The range space induced by $\fcal^c_{log} = \{c \cdot g(x_i \beta)\,|\, i\in [n]\}$ satisfies $\Delta(\mathfrak{R}_{\fcal^c_{log}}) \leq d + 1$.
\end{lemma}
\begin{lemma}
	\label{lem:vc2}
	Let $X\in\REALS^{n\times d}$ be weighted by $w\in\REALS^{n}$ where $w_i \in \{v_1,\ldots,v_t\}$ for all $i\in [n]$.
	The range space induced by $\fcal_{log} = \{w_i  \cdot g(x_i \beta) \mid i\in [n]\}$ satisfies $\Delta(\mathfrak{R}_{\fcal_{log}}) \leq t\cdot (d + 1)$.
\end{lemma}
We will see later how to bound the number of distinct weights $t$ by a logarithmic term in the range of the involved weights. It remains for us to derive tight and efficiently computable upper bounds on the sensitivities.

\textbf{Base Algorithm}
We show that sampling proportional to the \emph{square root} of the $\ell_2$-leverage scores augmented by $w_i/\sum_{j\in [n]} w_j$ yields a coreset whose size is roughly linear in $\mu$ and the dependence on the input size is roughly $\sqrt{n}$. In what follows, let $\wcal=\sum_{i\in [n]} w_i$.

We make a case distinction covered by lemmas \ref{lem:one} and \ref{lem:two}. The intuition in the first case is that for a sufficiently large positive entry $z$, we have that $|z|\leq g(z) \leq 2 |z|$. The lower bound holds even for all non-negative entries. Moreover, for $\mu$-complex inputs we are able to relate the $\ell_1$ norm of all entries to the positive ones, which will yield the desired bound, arguing similarly to the techniques of \cite{ClarksonW15_focs} though adapted here for logistic regression.
\begin{lemma}
	\label{lem:one}
	Let $X\in\REALS^{n\times d}$ weighted by $w\in\REALS_{> 0}^n$ be $\mu$-complex. Let $U$ be an orthonormal basis for the columnspace of $D_wX$. If for index $i$, the supreme $\beta$ in (\ref{eqn:sensitivities}) satisfies $0.5 \leq x_i \beta$ then $w_i g(x_i\beta) \leq 2(1+\mu)\twonorm{U_i}f_w(X\beta)$.
\end{lemma}
In the second case, the element under study is bounded by a constant. We consider two sub cases. If there are a lot of contributions, which are not too small, and thus cost at least a constant each, then we can lower bound the total cost by a constant times their total weight. If on the other hand there are many very small negative values, then this implies again that the cost is within a $\mu$ fraction of the total weight.
\begin{lemma}
	\label{lem:two}
	Let $X\in\REALS^{n\times d}$ weighted by $w\in\REALS_{> 0}^n$ be $\mu$-complex. If for index $i$, the supreme $\beta$ in (\ref{eqn:sensitivities}) satisfies $0.5 \geq x_i\beta$ then $w_i g(x_i\beta) \leq \frac{(20+\mu)w_i}{\wcal} f_w(X\beta)$.
\end{lemma}
Combining both lemmas yields general upper bounds on the sensitivities that we can use as an importance sampling distribution. We also derive an upper bound on the total sensitivity that will be used to bound the sampling complexity.
\begin{lemma}
	\label{lem:totalsensitivity}
	Let $X\in\REALS^{n\times d}$ weighted by $w\in\REALS_{> 0}^n$ be $\mu$-complex. Let $U$ be an orthonormal basis for the columnspace of $D_wX$. For each $i\in [n]$, the sensitivity of $g_i(\beta)=g(x_i\beta)$ for the weighted logistic regression function is bounded by $\varsigma_i \leq s_i = (20+2\mu)\cdot(\twonorm{U_i} + w_i/\wcal)$.
	The total sensitivity is bounded by $\mathfrak{S} \leq S \leq 44\mu \sqrt{nd}$. 
\end{lemma}
We combine the above results into the following theorem.
\begin{theorem}
	\label{thm:basicalg}
	Let $X\in\REALS^{n\times d}$ weighted by $w\in\REALS^n$ be $\mu$-complex. 
	Let $\omega = \frac{w_{\max}}{w_{\min}}$ be the ratio between the maximum and minimum weight in $w$.
	Let $\eps \in (0,1/2)$. There exists a $(1\pm\eps)$-coreset of $X,w$ for logistic regression of size $k\in O(\frac{\mu\sqrt{n}}{\eps^2} d^{3/2}\log(\mu n d)\log(\omega n))$. Such a coreset can be constructed in two passes over the data, in $O(\nnz(X)\log n+ \poly(d)\log n)$ time, and with success probability $1-1/n^c$ for any absolute constant $c>1$.
\end{theorem}

\textbf{Recursive Algorithm}
Here we develop a recursive algorithm, inspired by the recursive sampling technique of \cite{ClarksonW15} for the Huber $M$-estimator, though adapted here for logistic regression. This yields a better dependence on the input size. More specifically, we can diminish the leading $\sqrt{n}$ factor to only $\log^c(n)$ for an absolute constant $c$. One complication is that the parameter $\mu$ grows in the recursion, which we need to control, while another complication is having to deal with the separate $\ell_1$ and uniform parts of our sampling distribution. 

We apply the Algorithm of Theorem \ref{thm:basicalg} recursively. To do so, we need to ensure that after one stage of subsampling and reweighting, the resulting data set remains $\mu'$-complex for a value $\mu'$ that is not too much larger than $\mu$. To this end, we first bound the VC dimension of a range space induced by an $\ell_1$ related family of functions.
\begin{lemma}
	\label{lem:vc_ell1}
	The range space induced by $\fcal_{\ell_1} = \{h_i(\beta) = w_i|x_i\beta| \,|\,i\in [n] \}$ satisfies $\Delta(\mathfrak{R}_{\fcal_{\ell_1}}) \leq 10(d+1)$.
\end{lemma}
Applying Theorem \ref{thm:sensitivity} to $\fcal_{\ell_1}$ implies that the subsample of Theorem \ref{thm:basicalg} satisfies a so called $\eps$-subspace embedding property for $\ell_1$. Note that, by linearity of the $\ell_1$-norm, we can fold the weights into $D_wX$.
\begin{lemma}
	\label{lem:l1_subspace_embedding}
	Let $T$ be a sampling and reweighting matrix according to Theorem \ref{thm:basicalg}. I.e., $TD_wX$ is the resulting reweighted sample when Theorem \ref{thm:basicalg} is applied to $\mu$-complex input $X,w$. Then with probability $1-1/n^c$, for all $\beta\in\REALS^d$ simultaneously \[(1-\eps')\onenorm{D_wX\beta}\leq\onenorm{TD_wX\beta}\leq (1+\eps')\onenorm{D_wX\beta}\]
	holds, where $\eps'= \eps/\sqrt{\mu+1}$.
\end{lemma}
Using this, we can show that the $\mu$-complexity is not violated too much after one stage of sampling.
\begin{lemma}
	\label{lem:muerror}
	Let $T$ be a sampling and reweighting matrix according to Theorem \ref{thm:basicalg} where parameter $\eps$ is replaced by $\eps/\sqrt{\mu+1}$. That is $TD_wX$ is the resulting reweighted sample when Theorem \ref{thm:basicalg} succeeds on $\mu$-complex input $X,w$. Suppose that simultaneously Lemma \ref{lem:l1_subspace_embedding} holds. Let $$\mu' = \mu_{Tw}(X) = \sup\limits_{\beta\in \REALS^d} \frac{\onenorm{(TD_wX\beta)^+}}{\onenorm{(TD_wX\beta)^-}}.$$ Then we have $\mu'\leq(1+\eps)\mu.$
\end{lemma}
Now we are ready to prove our theorem regarding the recursive subsampling algorithm.
\begin{theorem}
	\label{thm:recalg}
	Let $X\in\REALS^{n\times d}$ be $\mu$-complex. Let $\eps \in (0,1/2)$. There exists a $(1\pm\eps)$-coreset of $X$ for logistic regression of size $k\in O(\frac{\mu^3}{\eps^4} d^{3}\log^2(\mu n d)\log^2 n \, (\log \log n)^4 )$. Such a coreset can be constructed in time $O((\nnz(X)+\poly(d))\log n\log\log n)$ in $2\log(\frac{1}{\eta})$ passes over the data for a small $\eta>0$, assuming the machine has access to sufficient memory to store and process $\tilde O(n^\eta)$ weighted points. The success probability is $1-1/n^c$ for any absolute constant $c>1$.
\end{theorem}

\section{Experiments}
We ran a series of experiments to illustrate the performance of our coreset method. All experiments were run on a Linux machine using an Intel i7-6700, 4 core CPU at 3.4~GHz, and 32GB of RAM. We implemented our algorithms in Python. Now, we compare our basic algorithm to simple uniform sampling and to sampling proportional to the sensitivity upper bounds given by \cite{HugginsCB16}.

\textbf{Implementation Details}
The approach of \cite{HugginsCB16} is based on a $k$-means++ clustering \citep{ArthurV07} on a small uniform sample of the data and was performed using standard parameters taken from the publication. For this purpose we used parts of their original Python code. However, we removed the restriction of the domain of optimization to a region of small radius around the origin. This way, we enabled unconstrained regression in the domain $\REALS^d$.

The exact QR-decomposition is rather slow on large data matrices. We thus optimized the running time of our approach in the following way. We used a fast approximation algorithm based on the sketching techniques of \cite{ClarksonW13}, cf. \citep{Woodruff14}. That leads to a provable constant approximation of the square root of the leverage scores with constant probability, cf. \citep{DrineasMMW12}, which means that the total sensitivity bounds given in our theory will grow by only a small constant factor. A detailed description of the algorithm is in the proof of Theorem \ref{thm:basicalg}.

The subsequent optimization was done for all approaches with the standard gradient based optimizer from the \texttt{scipy.optimize}\footnote{http://www.scipy.org/} package.

\textbf{Data Sets}
We briefly introduce the data sets that we used. 
The \textsc{Webb Spam}\footnote{https://www.cc.gatech.edu/projects/doi/WebbSpamCorpus.html} ~data consists of $350,000$ unigrams with $127$ features from web pages which have to be classified as spam or normal pages ($61\%$ positive).
The \textsc{Covertype}\footnote{https://archive.ics.uci.edu/ml/datasets/covertype} ~data consists of $581,012$ cartographic observations of different forests with $54$ features. The task is to predict the type of trees at each location ($49\%$ positive).
The \textsc{KDD Cup '99}\footnote{http://kdd.ics.uci.edu/databases/kddcup99/kddcup99.html} ~data comprises $494,021$ network connections with $41$ features and the task is to detect network intrusions ($20\%$ positive).

\textbf{Experimental Assessment}
For each data set we assessed the total running times for computing the sampling probabilities, sampling and optimizing on the sample.
In order to assess the approximation accuracy we examined the relative error $|\mathcal{L}(\beta^*|X)-\mathcal{L}(\tilde\beta|X)|/\mathcal{L}(\beta^*|X)$ of the negative log-likelihood for the maximum likelihood estimators obtained from the full data set $\beta^*$ and the subsamples $\tilde \beta$.

For each data set, we ran all three subsampling algorithms for a number of thirty regular subsampling steps in the range $k\in[\lfloor 2\sqrt{n}\rfloor ,\lceil n/16 \rceil]$. For each step, we present the mean relative error as well as the trade-off between mean relative error and running time, taken over twenty independent repetitions, in Figure \ref{fig:experiments}. Relative running times, standard deviations and absolute values are presented in Figure \ref{fig:supp:variances} respectively in Table \ref{tbl:abs_vals} in Appendix \ref{app:experiments}.

\begin{figure*}[ht!]
	\vskip 0.2in
	\begin{center}
		\begin{sc}
			\begin{tabular}{ccc}
				{\small\hspace{.5cm}Webb Spam}&{\small\hspace{.5cm}Covertype}&{\small\hspace{.5cm}KDD Cup '99} \\
				\includegraphics[width=0.30\linewidth]{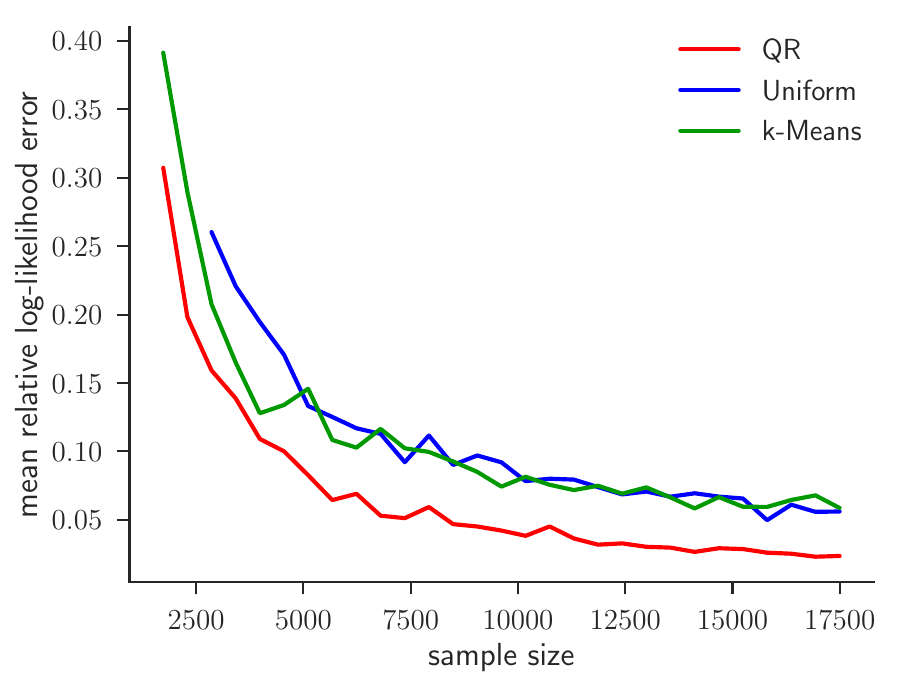}&
				\includegraphics[width=0.30\linewidth]{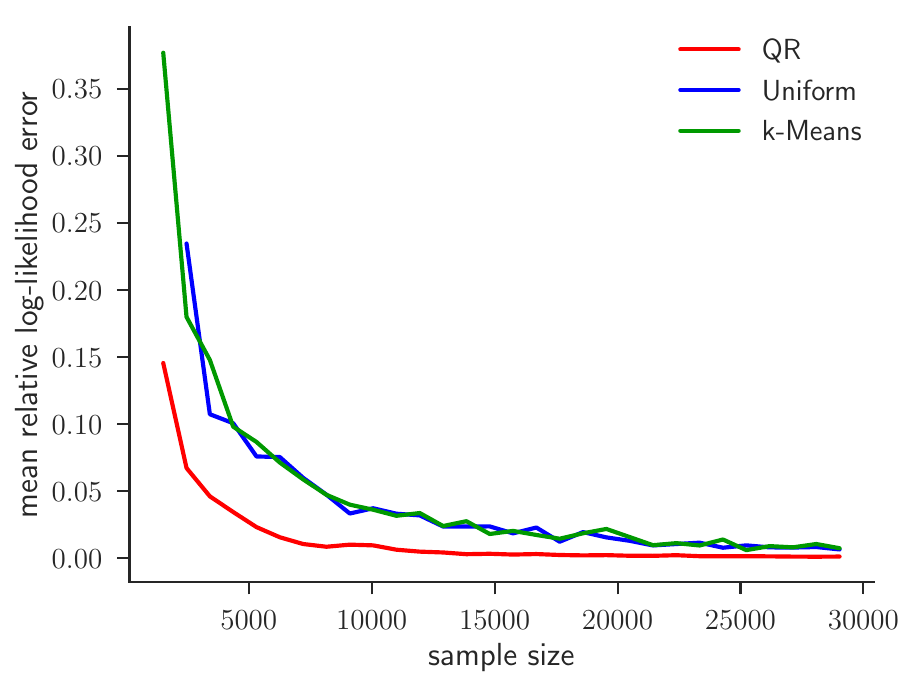}&
				\includegraphics[width=0.30\linewidth]{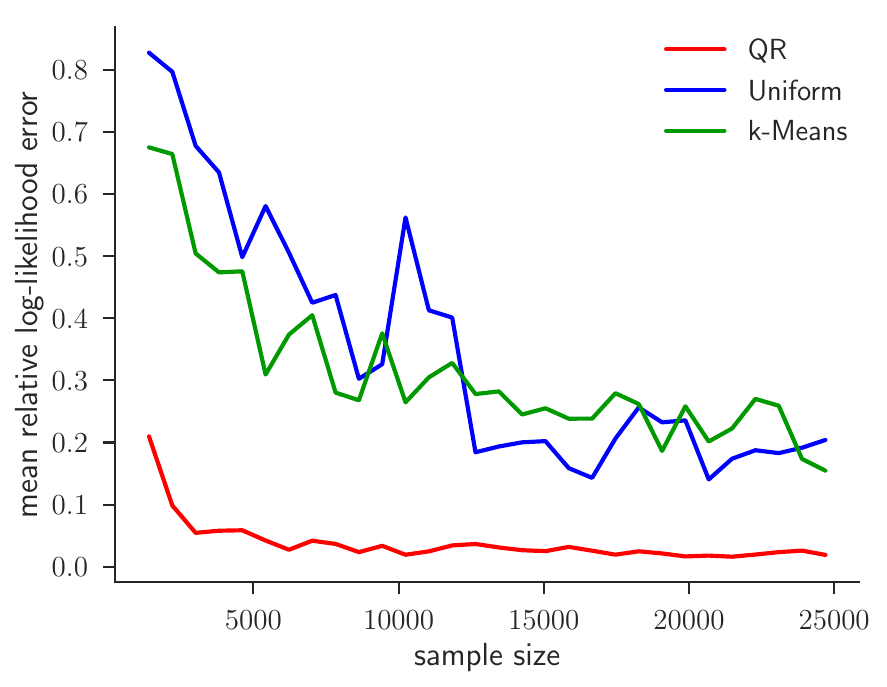} \\ 
				\includegraphics[width=0.30\linewidth]{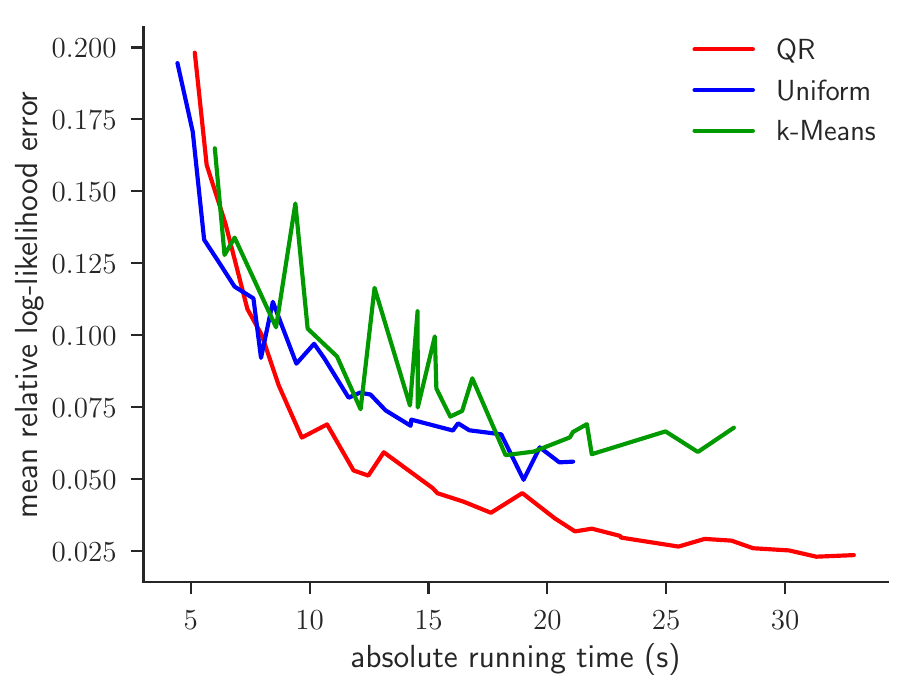}&
				\includegraphics[width=0.30\linewidth]{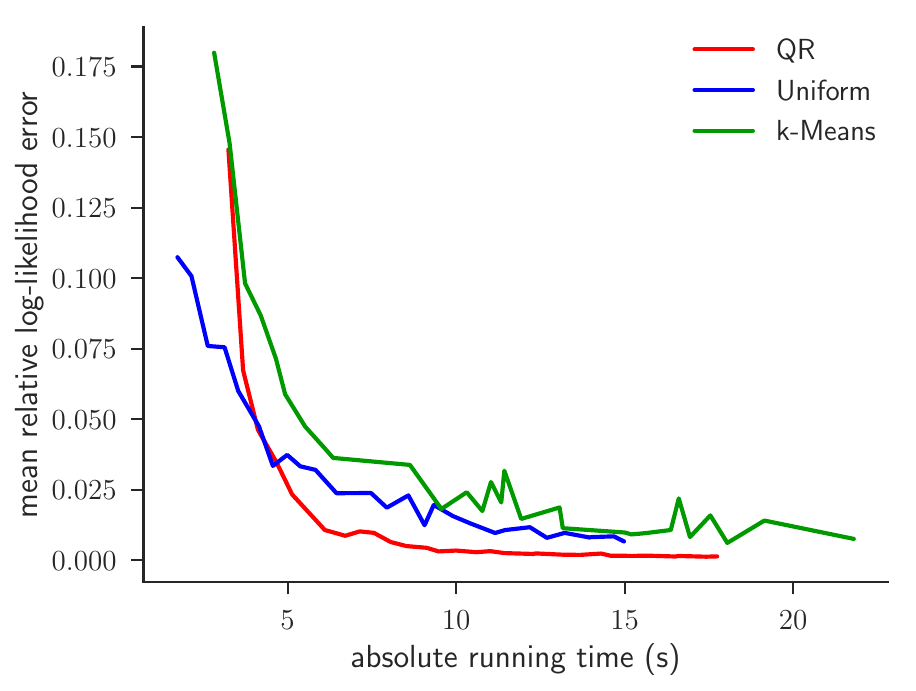}&
				\includegraphics[width=0.30\linewidth]{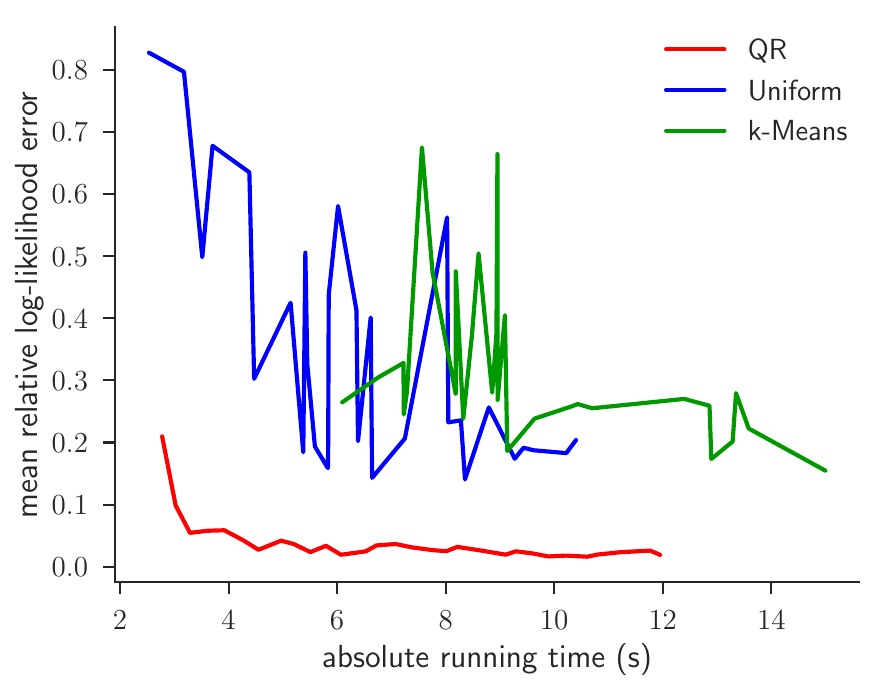} \\
			\end{tabular}
		\end{sc}
		\caption{Each column shows the results for one data set comprising thirty different coreset sizes (depending on the individual size of the data sets). The plotted values are means taken over twenty independent repetitions of each experiment. The plots in the upper row show the mean relative log-likelihood errors of the three subsampling distributions, uniform sampling (blue), our QR derived distribution (red), and the $k$-means based distribution (green). All values are relative to the corresponding optimal log-likelihood values of the optimization task on the full data set. The plots in the lower row show the trade-off between running time and relative errors (lower is better).}
		\label{fig:experiments}
	\end{center}
	\vskip -0.2in
\end{figure*}

\textbf{Evaluation}
The accuracy of the QR-sampling distribution outperforms uniform sampling and the distribution derived from $k$-means on all instances.
This is especially true for small sampling sizes. Here, the relative error especially for uniform sampling tends to deteriorate.
While $k$-means sampling occasionally improved over uniform sampling for small sample sizes, the behavior of both distributions was similar for larger sampling sizes. The standard deviations had a similarly low magnitude as the mean values, where the QR method usually showed the lowest values.

The trade-off between the running time and relative errors shows a common picture for \textsc{Webb Spam} and \textsc{Covertype}. QR is nearly always more accurate than the other algorithms for a similar time budget, except for regions where the relative error is large, say above 5-10\% while for larger time budgets, QR is better by a factor between $1.5$-$3$ and drops more quickly towards $0$. The conclusion so far could be that for a quick guess, say a $1.1$-approximation, the competitors are faster, but to provably obtain a reasonably small relative error below 5\%, QR outperforms its competitors. However, for \textsc{KDD Cup '99}, QR always has a lower error than its competitors. Their relative errors remain above 15\% or much worse, while QR never exceeds 22\% and drops quickly below 4\%. As a side note, our estimates for $\mu$ support our experimental findings, especially that \textsc{KDD Cup '99} seems more difficult to approximate than the others. The estimated values were $4.39$ for \textsc{Webb Spam}, $1.86$ for \textsc{Covertype}, and $35.18$ for \textsc{KDD Cup '99}.

The relative running time for the QR-distribution was comparable to $k$-means and only slightly higher than uniform sampling. However, it never exceeded a factor of two compared to its competitors and remained negligible compared to the full optimization task, see Figure \ref{fig:supp:variances} in Appendix \ref{app:experiments}.
The standard deviations were negligible except for the $k$-means algorithm and the \textsc{KDD Cup '99} data set, where the uniform and $k$-means based algorithms showed larger values. The QR method had much lower standard deviations. This indicates that the resulting coresets are more stable for the subsequent numerical optimization.

We note that the savings of all presented data reduction methods become even more significant when performing more time consuming data analysis tasks like MCMC sampling in a Bayesian setting, see e.g., \citep{HugginsCB16,GeppertIMQS17}.

\section{Conclusions}
We first showed that (sublinear) coresets for logistic regression do not exist in general. It is thus necessary to make further assumptions on the nature of the data. To this end we introduced a new complexity measure $\mu(X)$, which quantifies the amount of overlap of positive and negative classes and the balance in their cardinalities. We developed the first rigorously sublinear $(1\pm\eps)$-coresets for logistic regression, given that the original data has small $\mu$-complexity. The leading factor is $O(\eps^{-2}\mu\sqrt{n})$. We have further developed a recursive coreset construction that reduces the dependence on the input size to only $O(\log^c n)$ for absolute constant $c$. This comes at the cost of an increased dependence on $\mu$. However, it is beneficial for very large and well-behaved data. Our algorithms are space efficient, and can be implemented in a variety of models, used to tackle the challenges of large data sets, such as $2$-pass streaming, and massively parallel frameworks like Hadoop and MapReduce, and can be implemented to run in input sparsity time $\tilde O(\nnz(X))$, which is especially beneficial for sparsely encoded input data.

Our experimental evaluation shows that our implementation of the basic algorithm outperforms uniform sampling as well as state of the art methods in the area of coresets for logistic regression while being competitive to both regarding its running time.

\section*{Acknowledgments} We thank the anonymous reviewers for their valuable comments. We also thank our student assistant Moritz Paweletz for implementing and conducting the experiments. This work was partly supported by the German Science Foundation (DFG) Collaborative Research Center SFB 876 "Providing Information by Resource-Constrained Analysis", projects A2 and C4 and by the ERC Advanced Grant 788893 AMDROMA. 

\bibliography{references}

\begin{thebibliography}{44}
\providecommand{\natexlab}[1]{#1}
\providecommand{\url}[1]{\texttt{#1}}
\expandafter\ifx\csname urlstyle\endcsname\relax
  \providecommand{\doi}[1]{doi: #1}\else
  \providecommand{\doi}{doi: \begingroup \urlstyle{rm}\Url}\fi

\bibitem[Agarwal et~al.(2004)Agarwal, Har{-}Peled, and
  Varadarajan]{AgarwalHV04}
Agarwal, Pankaj~K., Har{-}Peled, Sariel, and Varadarajan, Kasturi~R.
\newblock Approximating extent measures of points.
\newblock \emph{Journal of the {ACM}}, 51\penalty0 (4):\penalty0 606--635,
  2004.

\bibitem[Alaoui \& Mahoney(2015)Alaoui and Mahoney]{AlM15}
Alaoui, Ahmed~El and Mahoney, Michael~W.
\newblock Fast randomized kernel ridge regression with statistical guarantees.
\newblock In \emph{Advances in Neural Information Processing Systems 28
  {(NIPS)}}, pp.\  775--783, 2015.

\bibitem[Arthur \& Vassilvitskii(2007)Arthur and Vassilvitskii]{ArthurV07}
Arthur, David and Vassilvitskii, Sergei.
\newblock k-means++: the advantages of careful seeding.
\newblock In \emph{Proceedings of the $18^{th}$ Annual {ACM-SIAM} Symposium on
  Discrete Algorithms {(SODA)}}, pp.\  1027--1035, 2007.

\bibitem[Bachem et~al.(2018)Bachem, Lucic, and Krause]{BachemLK18}
Bachem, Olivier, Lucic, Mario, and Krause, Andreas.
\newblock Scalable $k$-means clustering via lightweight coresets.
\newblock In \emph{Proceedings of the 24th {ACM} {SIGKDD} International
  Conference on Knowledge Discovery {\&} Data Mining {(KDD)}}, pp.\
  1119--1127, 2018.

\bibitem[Balcan et~al.(2015)Balcan, Manthey, R{\"o}glin, and
  Roughgarden]{BalcanMRR15}
Balcan, Marina-Florina, Manthey, Bodo, R{\"o}glin, Heiko, and Roughgarden, Tim.
\newblock {Analysis of algorithms beyond the worst case (Dagstuhl seminar
  14372)}.
\newblock \emph{Dagstuhl Reports}, 4\penalty0 (9):\penalty0 30--49, 2015.

\bibitem[Barger \& Feldman(2016)Barger and Feldman]{BargerF16}
Barger, Artem and Feldman, Dan.
\newblock $k$-means for streaming and distributed big sparse data.
\newblock In \emph{Proceedings of the {SIAM} International Conference on Data
  Mining {(SDM)}}, pp.\  342--350, 2016.

\bibitem[Blumer et~al.(1989)Blumer, Ehrenfeucht, Haussler, and
  Warmuth]{BlumerEHW89}
Blumer, Anselm, Ehrenfeucht, Andrzej, Haussler, David, and Warmuth, Manfred~K.
\newblock Learnability and the {V}apnik-{C}hervonenkis dimension.
\newblock \emph{Journal of the {ACM}}, 36\penalty0 (4):\penalty0 929--965,
  1989.

\bibitem[Braverman et~al.(2016)Braverman, Feldman, and Lang]{BravermanFL16}
Braverman, Vladimir, Feldman, Dan, and Lang, Harry.
\newblock New frameworks for offline and streaming coreset constructions.
\newblock \emph{arXiv preprint CoRR}, abs/1612.00889, 2016.

\bibitem[Chao(1982)]{Chao82}
Chao, M.~T.
\newblock A general purpose unequal probability sampling plan.
\newblock \emph{Biometrika}, 69\penalty0 (3):\penalty0 653--656, 1982.

\bibitem[Clarkson(2005)]{Clarkson05}
Clarkson, Kenneth~L.
\newblock Subgradient and sampling algorithms for $\ell_1$ regression.
\newblock In \emph{Proceedings of the $16^{th}$ annual ACM-SIAM symposium on
  Discrete algorithms {(SODA)}}, pp.\  257--266, 2005.

\bibitem[Clarkson \& Woodruff(2013)Clarkson and Woodruff]{ClarksonW13}
Clarkson, Kenneth~L. and Woodruff, David~P.
\newblock Low rank approximation and regression in input sparsity time.
\newblock In \emph{Symposium on Theory of Computing {(STOC)}}, pp.\  81--90,
  2013.

\bibitem[Clarkson \& Woodruff(2015{\natexlab{a}})Clarkson and
  Woodruff]{ClarksonW15}
Clarkson, Kenneth~L. and Woodruff, David~P.
\newblock Sketching for \emph{M}-estimators: {A} unified approach to robust
  regression.
\newblock In \emph{Proceedings of the $26^{th}$ Annual {ACM-SIAM} Symposium on
  Discrete Algorithms {(SODA)}}, pp.\  921--939, 2015{\natexlab{a}}.

\bibitem[Clarkson \& Woodruff(2015{\natexlab{b}})Clarkson and
  Woodruff]{ClarksonW15_focs}
Clarkson, Kenneth~L. and Woodruff, David~P.
\newblock Input sparsity and hardness for robust subspace approximation.
\newblock In \emph{{IEEE} 56th Annual Symposium on Foundations of Computer
  Science {(FOCS)}}, pp.\  310--329, 2015{\natexlab{b}}.

\bibitem[Clarkson et~al.(2016)Clarkson, Drineas, Magdon{-}Ismail, Mahoney,
  Meng, and Woodruff]{ClarksonDMMMW16}
Clarkson, Kenneth~L., Drineas, Petros, Magdon{-}Ismail, Malik, Mahoney,
  Michael~W., Meng, Xiangrui, and Woodruff, David~P.
\newblock The fast {C}auchy transform and faster robust linear regression.
\newblock \emph{{SIAM} J. Comput.}, 45\penalty0 (3):\penalty0 763--810, 2016.

\bibitem[Cohen et~al.(2015)Cohen, Lee, Musco, Musco, Peng, and
  Sidford]{CLMMPS15}
Cohen, Michael~B., Lee, Yin~Tat, Musco, Cameron, Musco, Christopher, Peng,
  Richard, and Sidford, Aaron.
\newblock Uniform sampling for matrix approximation.
\newblock In \emph{Proceedings of the Conference on Innovations in Theoretical
  Computer Science {(ITCS)}}, pp.\  181--190, 2015.

\bibitem[Cohen et~al.(2017)Cohen, Musco, and Musco]{CMM17}
Cohen, Michael~B., Musco, Cameron, and Musco, Christopher.
\newblock Input sparsity time low-rank approximation via ridge leverage score
  sampling.
\newblock In \emph{Proceedings of the $28^{th}$ Annual {ACM-SIAM} Symposium on
  Discrete Algorithms {(SODA)}}, pp.\  1758--1777, 2017.

\bibitem[Dasgupta et~al.(2009)Dasgupta, Drineas, Harb, Kumar, and
  Mahoney]{DasguptaDHKM09}
Dasgupta, Anirban, Drineas, Petros, Harb, Boulos, Kumar, Ravi, and Mahoney,
  Michael~W.
\newblock Sampling algorithms and coresets for $\ell_p$ regression.
\newblock \emph{{SIAM} Journal on Computing}, 38\penalty0 (5):\penalty0
  2060--2078, 2009.

\bibitem[Drineas et~al.(2006)Drineas, Mahoney, and Muthukrishnan]{DrineasMM06}
Drineas, Petros, Mahoney, Michael~W., and Muthukrishnan, S.
\newblock Sampling algorithms for $\ell_2$ regression and applications.
\newblock In \emph{Proceedings of the $17^{th}$ Annual {ACM-SIAM} Symposium on
  Discrete Algorithms {(SODA)}}, pp.\  1127--1136, 2006.

\bibitem[Drineas et~al.(2008)Drineas, Mahoney, and Muthukrishnan]{DrineasMM08}
Drineas, Petros, Mahoney, Michael~W., and Muthukrishnan, S.
\newblock Relative-error {CUR} matrix decompositions.
\newblock \emph{{SIAM} Journal on Matrix Analysis and Applications},
  30\penalty0 (2):\penalty0 844--881, 2008.

\bibitem[Drineas et~al.(2012)Drineas, Magdon{-}Ismail, Mahoney, and
  Woodruff]{DrineasMMW12}
Drineas, Petros, Magdon{-}Ismail, Malik, Mahoney, Michael~W., and Woodruff,
  David~P.
\newblock Fast approximation of matrix coherence and statistical leverage.
\newblock \emph{Journal of Machine Learning Research}, 13:\penalty0 3475--3506,
  2012.

\bibitem[Feldman \& Langberg(2011)Feldman and Langberg]{FeL11}
Feldman, Dan and Langberg, Michael.
\newblock A unified framework for approximating and clustering data.
\newblock In \emph{Proceedings of the 43rd {ACM} Symposium on Theory of
  Computing {(STOC)}}, pp.\  569--578, 2011.

\bibitem[Feldman et~al.(2011)Feldman, Faulkner, and Krause]{FeldmanFK11}
Feldman, Dan, Faulkner, Matthew, and Krause, Andreas.
\newblock Scalable training of mixture models via coresets.
\newblock In \emph{Advances in Neural Information Processing Systems 24
  {(NIPS)}}, pp.\  2142--2150, 2011.

\bibitem[Feldman et~al.(2013)Feldman, Schmidt, and Sohler]{FeldmanSS13}
Feldman, Dan, Schmidt, Melanie, and Sohler, Christian.
\newblock Turning big data into tiny data: Constant-size coresets for
  \emph{k}-means, {PCA} and projective clustering.
\newblock In \emph{Proceedings of the $24^{th}$ Annual {ACM-SIAM} Symposium on
  Discrete Algorithms {(SODA)}}, pp.\  1434--1453, 2013.

\bibitem[Geppert et~al.(2017)Geppert, Ickstadt, Munteanu, Quedenfeld, and
  Sohler]{GeppertIMQS17}
Geppert, Leo~N., Ickstadt, Katja, Munteanu, Alexander, Quedenfeld, Jens, and
  Sohler, Christian.
\newblock Random projections for {B}ayesian regression.
\newblock \emph{Statistics and Computing}, 27\penalty0 (1):\penalty0 79--101,
  2017.

\bibitem[Golub \& van Loan(2013)Golub and van Loan]{GolubL96}
Golub, Gene~H. and van Loan, Charles~F.
\newblock \emph{Matrix computations {(4.} ed.)}.
\newblock J. Hopkins Univ. Press, 2013.

\bibitem[Heinze \& Schemper(2002)Heinze and Schemper]{HeinzeS02}
Heinze, Georg and Schemper, Michael.
\newblock A solution to the problem of separation in logistic regression.
\newblock \emph{Statistics in Medicine}, 21\penalty0 (16):\penalty0 2409--2419,
  2002.

\bibitem[Huggins et~al.(2016)Huggins, Campbell, and Broderick]{HugginsCB16}
Huggins, Jonathan~H., Campbell, Trevor, and Broderick, Tamara.
\newblock Coresets for scalable {B}ayesian logistic regression.
\newblock In \emph{Advances in Neural Information Processing Systems 29
  {(NIPS)}}, pp.\  4080--4088, 2016.

\bibitem[Johnson \& Lindenstrauss(1984)Johnson and Lindenstrauss]{JohnsonL84}
Johnson, William~B and Lindenstrauss, Joram.
\newblock Extensions of {L}ipschitz mappings into a {H}ilbert space.
\newblock \emph{Contemporary Mathematics}, {26}\penalty0 (1):\penalty0
  189--206, 1984.

\bibitem[Kearns \& Vazirani(1994)Kearns and Vazirani]{KearnsV94}
Kearns, Michael~J. and Vazirani, Umesh~V.
\newblock \emph{An Introduction to Computational Learning Theory}.
\newblock {MIT} Press, 1994.

\bibitem[Kremer et~al.(1999)Kremer, Nisan, and Ron]{KNR99}
Kremer, Ilan, Nisan, Noam, and Ron, Dana.
\newblock On randomized one-round communication complexity.
\newblock \emph{Computational Complexity}, 8\penalty0 (1):\penalty0 21--49,
  1999.

\bibitem[Langberg \& Schulman(2010)Langberg and Schulman]{LangbergS10}
Langberg, Michael and Schulman, Leonard~J.
\newblock Universal $\varepsilon$-approximators for integrals.
\newblock In \emph{Proceedings of the $21^{st}$ Annual {ACM-SIAM} Symposium on
  Discrete Algorithms {(SODA)}}, pp.\  598--607, 2010.

\bibitem[Li et~al.(2013)Li, Miller, and Peng]{LiMP13}
Li, Mu, Miller, Gary~L., and Peng, Richard.
\newblock Iterative row sampling.
\newblock In \emph{54th Annual {IEEE} Symposium on Foundations of Computer
  Science {(FOCS)}}, pp.\  127--136, 2013.

\bibitem[Lucic et~al.(2016)Lucic, Bachem, and Krause]{LucicBK16}
Lucic, Mario, Bachem, Olivier, and Krause, Andreas.
\newblock Strong coresets for hard and soft {B}regman clustering with
  applications to exponential family mixtures.
\newblock In \emph{Proceedings of the 19th International Conference on
  Artificial Intelligence and Statistics {(AISTATS)}}, pp.\  1--9, 2016.

\bibitem[McCullagh \& Nelder(1989)McCullagh and Nelder]{McCullaghN89}
McCullagh, P. and Nelder, J.~A.
\newblock \emph{Generalized Linear Models}.
\newblock Chapman \& Hall, London, 1989.

\bibitem[Mehta \& Patel(1995)Mehta and Patel]{MethaP95}
Mehta, Cyrus~R. and Patel, Nitin~R.
\newblock Exact logistic regression: Theory and examples.
\newblock \emph{Statistics in Medicine}, 14\penalty0 (19):\penalty0 2143--2160,
  1995.

\bibitem[Molina et~al.(2018)Molina, Munteanu, and Kersting]{MolinaMK17}
Molina, Alejandro, Munteanu, Alexander, and Kersting, Kristian.
\newblock Core dependency networks.
\newblock In \emph{Proceedings of the 32nd {AAAI} Conference on Artificial
  Intelligence (AAAI)}, 2018.

\bibitem[Musco \& Musco(2017)Musco and Musco]{MM17}
Musco, Cameron and Musco, Christopher.
\newblock Recursive sampling for the {N}ystr{\"o}m method.
\newblock In \emph{Advances in Neural Information Processing Systems 30
  {(NIPS)}}, pp.\  3836--3848, 2017.

\bibitem[Reddi et~al.(2015)Reddi, P{\'{o}}czos, and Smola]{ReddiPS15}
Reddi, Sashank~J., P{\'{o}}czos, Barnab{\'{a}}s, and Smola, Alexander~J.
\newblock Communication efficient coresets for empirical loss minimization.
\newblock In \emph{Proceedings of the Thirty-First Conference on Uncertainty in
  Artificial Intelligence {(UAI)}}, pp.\  752--761, 2015.

\bibitem[Roughgarden(2017)]{Roughgarden17}
Roughgarden, Tim.
\newblock Beyond worst-case analysis, 2017.
\newblock \textit{Invited talk held at the {H}ighlights of {A}lgorithms
  conference {(HALG),} }2017.

\bibitem[Sohler \& Woodruff(2011)Sohler and Woodruff]{SohlerW11}
Sohler, Christian and Woodruff, David~P.
\newblock Subspace embeddings for the {$L_1$}-norm with applications.
\newblock In \emph{Proceedings of the 43rd {ACM} Symposium on Theory of
  Computing {(STOC)}}, pp.\  755--764, 2011.

\bibitem[Tolochinsky \& Feldman(2018)Tolochinsky and Feldman]{FeldmanT18}
Tolochinsky, Elad and Feldman, Dan.
\newblock Coresets for monotonic functions with applications to deep learning.
\newblock \emph{CoRR}, abs/1802.07382, 2018.

\bibitem[Vapnik(1995)]{Vapnik95}
Vapnik, Vladimir~N.
\newblock \emph{The Nature of Statistical Learning Theory}.
\newblock Springer, New York, USA, 1995.

\bibitem[Woodruff(2014)]{Woodruff14}
Woodruff, David~P.
\newblock Sketching as a tool for numerical linear algebra.
\newblock \emph{Foundations and Trends in Theoretical Computer Science},
  10\penalty0 (1-2):\penalty0 1--157, 2014.

\bibitem[Woodruff \& Zhang(2013)Woodruff and Zhang]{WoodruffZ13}
Woodruff, David~P. and Zhang, Qin.
\newblock Subspace embeddings and $\ell_p$-regression using exponential random
  variables.
\newblock In \emph{The 26th Conference on Learning Theory {(COLT)}}, pp.\
  546--567, 2013.

\end{thebibliography}
\bibliographystyle{icml2018}%TODO: remove

\clearpage
\appendix
\section{Proofs}
\label{app:proofs}
\begin{proof}[Proof of Theorem \ref{thm:muLP}.]
	Let $A = D_wX$. In $O(nd\log d + \poly(d))$ time, find an $\ell_1$-well-conditioned basis \citep{ClarksonDMMMW16} $U\in\REALS^{n\times d}$ of $A$, such that \[\forall \beta\in\REALS^d\colon \onenorm{\beta} \leq \onenorm{U\beta} \leq \poly(d) \onenorm{\beta}.\] 
	Then $\mu(U)$ and $\mu(A)$ are the same since $U$ and $A$ span the same columnspace. By linearity it suffices to optimize over unit-$\ell_1$ vectors $\beta$. If we minimize $\onenorm{(U\beta)^-}$ over unit-$\ell_1$ vectors $\beta$, and $t$ is the minimum value, then $\mu$ is at most $\poly(d)/t$, and at least $1/t$ by the well-conditioned basis property, so we just need to find $t$, which can be done with the following linear program:
	\begin{align*}
	\min\quad&\sum\nolimits_{i=1}^n b_i \\
	\operatorname{s.t.}\quad&\forall i\in[n]\colon (U\beta)_i = a_i - b_i \\
	&\forall i\in[d]\colon \beta_i = c_i - d_i \\
	&\sum\nolimits_{i=1}^d c_i + d_i \geq 1\\
	&\forall i\in[n]\colon a_i, b_i \geq 0\\
	&\forall i\in[d]\colon c_i, d_i \geq 0
	\end{align*}
	Note that $\sum\nolimits_{i=1}^d c_i + d_i \geq 1$ ensures $\onenorm{\beta} \geq 1$, but to minimize the objective function, one will always have $\onenorm{\beta}$. Further, if both $a_i$ and $b_i$ are positive for some $i$, they can both be reduced, reducing the objective function. So $\sum\nolimits_{i=1}^n b_i$ exactly corresponds to the minimum over $\beta\in\REALS^d$ of $\onenorm{(U\beta)^-}$.
\end{proof}

\begin{proof}[Proof of Theorem \ref{LB:streaming}.]
	Assume we had a streaming algorithm using $o(n/\log n)$ space. We construct the following protocol for I{\footnotesize NDEX}:
	Consider an instance of I{\footnotesize NDEX}, i.e., Alice has a string $x\in\{0,1\}^n$ and Bob has an index $i\in[n]$.
	We transform the instance into an instance for logistic regression.	For each $x_j=1$, Alice adds a point $p_j=(\cos (\frac{j}{n}),\sin (\frac{j}{n}))$. Note that all of these points have unit Euclidean norm and hence any single point may be linearly separated from the others. All of Alice's points have label $1$.
	Alice summarizes the point set by running the streaming algorithm and sends a message containing the working memory of the streaming algorithm to Bob. 
	Bob now adds the point $p_i=(1-\delta)\cdot(\cos (\frac{i}{n}),\sin (\frac{i}{n}))$ for small enough $\delta>0$ with label $-1$. 
	From the contents of Alice's message and $p_i$, Bob now obtains a solution to the logistic regression instance.
	Clearly, if Alice added $p_i$ and hence $x_i=1$ then the optimal solution will have cost at least $\ln(2)$, since there will be at least one misclassification.
	If, on the other hand, Alice did not add $p_i$ and hence $x_i=0$, then the two point sets are linearly separable and the cost tends to $0$.
	Distinguishing between these two cases, i.e. approximating the cost of logistic regression beyond a factor $\lim\limits_{x\rightarrow 0}\frac{\ln(2)}{x}$ solves the I{\footnotesize NDEX} problem.
	
	To conclude the theorem, let us consider the space required to encode the points added by Alice. For the reduction to work, it is only important that any point added by Alice can be linearly separated from the others. This can be achieved by using $O(\log n)$ bits per point, i.e., the space of Alice's point set is at most $n'\in O(n\log n)$. The space bound now follows from the lower bound of $\Omega(n) \subseteq \Omega(n'/\log n)$ bits due to \cite{KNR99} for the I{\footnotesize NDEX} problem.    
\end{proof}

\begin{proof}[Proof of Corollary \ref{LB:coreset}.]
	If we had a coreset construction with $o(n/\log n)$ points, we have a protocol for I{\footnotesize NDEX}: Alice computes a coreset for her point set defined in the proof of Theorem \ref{LB:streaming} and sends it to Bob. Bob computes an optimal solution on the union of the coreset and his point. This solves I{\footnotesize NDEX} using $o(n)$ communication, which contradicts the lower bound of \cite{KNR99}. So Alice's coreset cannot exist.
\end{proof}

\begin{proof}[Proof of Lemma \ref{lem:vc1}.]{\it (cf. \cite{HugginsCB16})}
	For all $G\subseteq\fcal^c_{log}$, we have
	\begin{align*}
	&|\{ G \cap R \mid R \in \ranges(\fcal^c_{log})\}| = |\{ \rng{G}(\beta,r) \mid \beta\in\REALS^d, r\in\REALS_{\geq 0} \}|
	\end{align*}
	Note that $g$ is invertible and monotone. Also note that $g^{-1}$ maps $\REALS_{\geq 0}$ surjectively into $\REALS$. For all $\beta\in\REALS^d, r\in \REALS_{\geq 0}$ we thus have 
	\begin{align*}
	\rng{G}(\beta,r)&=\{c\cdot g_i\in G \mid c\cdot g_i(\beta) \geq r\}\\
	 &= \{c\cdot g_i\in G \mid c\cdot g(x_i\beta) \geq r\} = \{c\cdot g_i\in G \mid x_i\beta \geq g^{-1}(r/c)\}.
	\end{align*}
	Now note that $\{c\cdot g_i\in G \mid x_i\beta \geq g^{-1}(r/c) \}$ corresponds to the set of points that is shattered by the affine hyperplane classifier $x_i \mapsto \mathbf{1}_{\{x_i\beta - g^{-1}(r/c) \geq 0\}}$.
	We can conclude that
	\begin{align*}
	&\left| \{\rng{G}(\beta,r) \mid \,\beta \in \REALS^d, r\in\REALS_{\geq 0}\} \right|=\left| \{\{g_i\in G \mid x_i\beta - s \geq 0 \} \mid \beta\in\REALS^d, s\in \REALS \}  \right|
	\end{align*}
	which means that the VC dimension of $\mathfrak{R}_{\fcal^c_{log}}$ is $d+1$ since the VC dimension of the set of hyperplane classifiers is $d+1$ \citep{KearnsV94,Vapnik95}.
\end{proof}

\begin{proof}[Proof of Lemma \ref{lem:vc2}.]
	We partition the functions into $t$ disjoint classes having equal weights. Let $F_i = \{ w_j\cdot g_j\in \fcal_{log} \mid w_j=v_i \}$, for $i\in [t]$. For the sake of contradiction, suppose $\Delta(\mathfrak{R}_{\fcal_{log}}) > t\cdot (d+1)$. Then there exists a set $G$ of size $|G| > t\cdot (d+1)$ that is shattered by the ranges of $\mathfrak{R}_{\fcal_{log}}$. Now consider the sets $F_i \cap G$, for $i\in [t]$. Due to the disjointness property, each set $F_i \cap G$ must be shattered by the ranges induced by $F_i$. But at least one of them must be as large as $\frac{|G|}{t}>\frac{t\cdot (d+1)}{t}=d+1$, which contradicts Lemma \ref{lem:vc1}. Thus $\Delta(\mathfrak{R}_{\fcal_{log}}) \leq t\cdot (d+1) \in O(dt)$ follows.
\end{proof}

\begin{proof}[Proof of Lemma \ref{lem:one}.]
	Let $D_wX=UR$, where $U$ is an orthonormal basis for the columnspace of $D_wX$. It follows from $0.5 \leq x_i\beta$ and monotonicity of $g$ that
	\begin{align*}
	w_i g(x_i\beta) &= w_i \g{\frac{w_i x_i \beta}{w_i}} = w_i \g{\frac{U_i R \beta}{w_i}}  
	\leq w_i \g{\frac{\twonorm{U_i}\twonorm{R\beta}}{w_i}} = w_i \g{\frac{\twonorm{U_i}\twonorm{UR\beta}}{w_i}} \\ 
	&= w_i \g{\frac{\twonorm{U_i}\twonorm{D_wX\beta}}{w_i}} \leq w_i\frac{2}{w_i}\twonorm{U_i} \twonorm{D_wX\beta} 
	\leq 2\twonorm{U_i} \onenorm{D_wX\beta} \\
	&\leq 2\twonorm{U_i} (1+\mu)\onenorm{(D_wX\beta)^+} 
	= 2\twonorm{U_i} (1+\mu) \sum_{j:w_jx_j\beta \geq 0} w_j |x_j\beta| \\
	&\leq 2\twonorm{U_i} (1+\mu) \sum_{j:x_j\beta \geq 0} w_j g(x_j\beta) 
	\leq 2\twonorm{U_i} (1+\mu)f_w(X\beta).\qedhere 
	\end{align*}
\end{proof}

\begin{proof}[Proof of Lemma \ref{lem:two}.]
	Let $K^-= \{j\in [n]\;|\;x_j\beta \leq -2 \}$ and $K^+=\{j\in [n]\;|\;x_j\beta > -2 \}$. Note that $g(-2)>1/10$ and $g(x_i\beta) \leq g(0.5)<1$. Also, $\sum_{j\in K^-} w_j + \sum_{j\in K^+} w_j=\wcal.$
	
	Thus if $\sum_{j\in K^+} w_j \geq \frac{1}{2} \wcal$ then 
	\begin{align*}
	f_w(X\beta) &= \sum\nolimits_{i=1}^n w_j g(x_j\beta) \geq \frac{\sum_{j\in [n]} w_j}{20} \geq \frac{\wcal}{20w_i} \cdot w_i g(x_i\beta).
	\end{align*}
	
	If on the other hand $\sum_{j\in K^+} w_j < \frac{1}{2} \wcal$ then $\sum_{j\in K^-} w_j \geq \frac{1}{2} \wcal$. Thus 
	\begin{align*}
	f_w & (X\beta) \geq \onenorm{(D_wX\beta)^+} \geq {\onenorm{(D_wX\beta)^-}}/{\mu} \geq  \left(2\cdot \frac{\sum_{j\in [n]} w_j}{2}\right)\Big/{\mu} \geq \frac{\wcal}{\mu w_i} \cdot w_i g(x_i\beta). \qedhere
	\end{align*}
\end{proof}

\begin{proof}[Proof of Lemma \ref{lem:totalsensitivity}.]
	From Lemma \ref{lem:one} and Lemma \ref{lem:two} we have for each $i$
	\begin{align*}
	\varsigma_i = \sup_\beta \frac{w_i g(x_i\beta)}{f_w(X\beta)}
	&\leq 2(1+\mu)\twonorm{U_i} + (20+\mu)\frac{w_i}{\wcal} \leq (20+2\mu) \Big( \twonorm{U_i} + \frac{w_i}{\wcal} \Big)
	\end{align*}
	From this, the second claim follows via the Cauchy-Schwarz inequality and using the fact that the Frobenius norm satisfies $\Fnorm{U}=\sqrt{\sum\nolimits_{i\in[n],j\in[d]} |U_{ij}|^2}=\sqrt{d}$ due to orthonormality of $U$. We have
	\begin{align*}
	\mathfrak{S} = \sum\nolimits_{i=1}^n \varsigma_i
	&\leq (20+2\mu) \sum\nolimits_{i=1}^n \Big( \twonorm{U_i} + \frac{w_i}{\wcal} \Big) \leq 22\mu (\sqrt{n}\Fnorm{U} + 1) \leq 44\mu \sqrt{nd}\,.\qedhere
	\end{align*}
\end{proof}

\begin{proof}[Proof of Theorem \ref{thm:basicalg}.]\label{thm:basicalg:proof}
	The algorithm computes the QR-decomposition $D_wX=QR$ of $D_wX$. Note that $Q$ is an orthonormal basis for the columnspace of $D_wX$. We would like to use the upper bounds on the sensitivities from Lemma \ref{lem:totalsensitivity}. Namely, to sample the input points proportional to the sampling probabilities $\frac{s_i}{\sum\nolimits_{j=1}^n s_j} = \frac{\twonorm{Q_i}+w_i/\wcal}{\sum\nolimits_{j=1}^n (\twonorm{Q_j}+w_j/\wcal)}.$ However, to keep control of the VC dimension of the involved range space, we modify them to obtain upper bounds $s_i'$ such that each value ${s_i'}/{w_i}$ corresponds to ${s_i}/{w_i}$ but is rounded up to the closest power of two. It thus holds $s_i \leq s'_i\leq 2 s_i$ for all $i\in [n]$. The input points are sampled proportional to the sampling probabilities $p_i = {s'_i}/{\sum\nolimits_{j=1}^n s'_j}.$ From Lemma \ref{lem:totalsensitivity} we know that $S'=\sum\nolimits_{j=1}^n s'_j \leq 2 S \in O(\mu\sqrt{nd})$.
	
	In the proof of Theorem \ref{thm:sensitivity}, the VC dimension bound is applied to a set of functions which are reweighted by $\frac{S'w_i}{s'_ik}$. We denote this set of functions $\fcal_{log}$. Now note that the sensitivities satisfy
	\begin{align}
	\label{eqn:sensitivitypowers}
	\frac{2}{w_{\min}} \geq \frac{2}{w_i} \geq \frac{s_i'}{w_i} \geq \frac{s_i}{w_i} \geq \sup_\beta \frac{g(x_i\beta)}{\sum\nolimits_{j=1}^{n} w_j g(x_j\beta)} \overset{\beta = 0}{\geq} \frac{1}{\sum\nolimits_{j=1}^{n} w_j} \geq \frac{1}{{n} w_{\max}} \,.
	\end{align}
	Also note that $k$ and $S'$ are fixed values. Since the values ${s_i'}/{w_i}$ are scaled to powers of two, by (\ref{eqn:sensitivitypowers}) there can be at most $O(\log \frac{nw_{\max}}{w_{\min}})\subseteq O(\log (\omega n))$ distinct values of $\frac{S'w_i}{s'_ik}$. Putting this into Lemma \ref{lem:vc2}, we have $\Delta(\mathfrak{R}_{\fcal_{log}}) \in O(d\log (\omega n))$.
	
	Putting all these pieces into Theorem \ref{thm:sensitivity} for error parameter $\eps \in (0,1/2)$ and failure probability $\eta = n^{-c}$, we have that a reweighted random sample of size
	\begin{align*}
	k &\in O\left(\frac{S'}{\eps^2}\left( \Delta(\mathfrak{R}_{\fcal_{log}}) \log S' + \log \left(\frac{1}{\eta}\right) \right) \right) \\
	&\subseteq O\left(\frac{\mu\sqrt{nd}}{\eps^2}\left( d \log(\mu\sqrt{nd}) \log (\omega n) + \log \left(n^c\right) \right) \right)\\
	&\subseteq O\left(\frac{\mu\sqrt{n}}{\eps^2} d^{3/2} \log(\mu nd)\log (\omega n) \right)
	\end{align*}
	is a $(1\pm\eps)$ coreset with probability $1-1/n^c$ as claimed.	
	
	It remains to prove the claims regarding streaming and running time. We can compute the QR-decomposition of $D_wX$ in time $O(nd^2)$, see \citep{GolubL96}. Once $Q$ is available, we can inspect it row-by-row computing $\twonorm{Q_i}+w_i/\wcal$ and give it as input together with $x_i$ to $k$ independent copies of a weighted reservoir sampler \citep{Chao82}, which takes $O(\nnz(X))$ time to collect all sampled non-zero entries. This gives a total running time of $O(nd^2)$ since the computations are dominated by the QR-decomposition.
	
	We argue how to implement the first step in one streaming pass over the data in time $O(\nnz(X)\log n + \poly(d))$. Using the sketching techniques of \cite{ClarksonW13}, cf. \citep{Woodruff14}, we can obtain a provably constant approximation of the square root of the leverage scores $\twonorm{Q_i}$ with constant probability \citep{DrineasMMW12}. This means that the total sensitivity bound $S$ grows only by a small constant factor and does not affect the asymptotic analysis presented above. The idea is to first sketch the data matrix $X\in\REALS^{n\times d}$ to a significantly smaller matrix $\tilde{X}\in\REALS^{n'\times d}$, where $n'\in O(d^2)$. This takes only $O(\nnz(X)\log n + \poly(d)\log n)$ time, where the $\poly(d)$ and $\log n$ factors are only needed to amplify the success probability from constant to $\frac{1}{n^c}$ \citep{Woodruff14}. Performing the QR-decomposition $\tilde{X}=\tilde{Q}\tilde{R}$ takes $O(n'd^2)\subseteq O(d^4)$ time.
	
	Now, to compute a fast approximation to the row norms, we use a Johnson-Lindenstrauss transform, i.e., a matrix $G\in\REALS^{d\times m}, m\in O(\log n)$, whose entries are i.i.d. $G_{ij} \sim N(0,\frac{1}{m})$ \citep{JohnsonL84}. We compute the approximation to the row norms used in our sampling probabilities in a second pass over the data, as $\twonorm{\tilde{U_i}}=\twonorm{ X_i(\tilde{R}^{-1}G) }$, for $i\in [n]$. As we do so, we can feed these augmented with the corresponding weight directly to the reservoir sampler. The latter is a streaming algorithm itself and updates its sample in constant time. The matrix product $\tilde{R}^{-1}G$ takes at most $O(d^2\log n)$ time, and the streaming pass can be done in $O(\nnz(X)\log n)$.
	
	This sums up to two passes over the data and a running time of $O(\nnz(X)\log n+ \poly(d)\log n)$.
\end{proof}

\begin{proof}[Proof of Lemma \ref{lem:vc_ell1}.]
	Fix an arbitrary $G\subseteq \fcal_{\ell_1}$. Let $\Omega=\REALS^d\times\REALS_{\geq0}$. We attempt to bound the quantity
	\begin{align}
	|\{ G \cap R \mid & \,R \in \ranges(\fcal_{\ell_1})\}| \notag\\
	&=| \{\rng{G}(\beta,r)\mid \beta\in \REALS^d, r\in\REALS_{\geq 0}\} | \notag \\\notag
	&= | \bigcup\limits_{(\beta,r)\in\Omega}\{\{h_i\in G \mid h_i(\beta) \geq r \}\} | \\\notag
	&= | \bigcup\limits_{(\beta,r)\in\Omega} \{\{h_i\in G \mid w_i x_i\beta \geq r \vee - w_i x_i\beta \geq r\} \} | \\\notag
	&\leq \left|  \bigcup\limits_{(\beta,r)\in\Omega} \{\{h_i\in G \mid w_i x_i\beta \geq r \}\}  \right|  \cdot \left| \bigcup\limits_{(\beta,r)\in\Omega} \{\{h_i\in G \mid -w_i x_i\beta \geq r \}\} \right| \\
	&= \left|  \bigcup\limits_{(\beta,r)\in\Omega} \{\{h_i\in G \mid w_i x_i\beta \geq r \}\}  \right|^2. \label{eqn:rngspace}
	\end{align}
	The inequality holds, since each non-empty set in the collection on the LHS satisfies either of the conditions of the sets in the collections on the RHS, or both, and is thus the union of two of those sets, one from each collection. It can thus comprise at most all unions obtained from combining any two of these sets.
	The last equality holds since for each fixed $\beta$ we also union over $-\beta$ as we reach over all $\beta\in\REALS^d$. The two sets are thus equal.
	
	Now note that each set $\{h_i\in G \mid w_i x_i\beta \geq r \}$ equals the set of weighted points that is shattered by the affine hyperplane classifier $w_ix_i \mapsto \mathbf{1}_{\{w_ix_i\beta - r \geq 0\}}$. Note that the VC dimension of the set of hyperplane classifiers is $d+1$ \citep{KearnsV94,Vapnik95}. To conclude the claimed bound on $\Delta(\mathfrak{R}_{\fcal_{\ell_1}})$ it is sufficient to show that the above term (\ref{eqn:rngspace}) is bounded strictly below $2^{|G|}$ for $|G|=10(d+1)$. By a bound given in \citep{BlumerEHW89,KearnsV94} we have for this particular choice
	\begin{align*}
	(\ref{eqn:rngspace}) &\leq \left| \{\{h_i\in G \mid w_i x_i\beta - r \geq 0 \} \mid \beta\in\REALS^d, r\in \REALS \}  \right|^2 \leq \left(\frac{e|G|}{d+1}\right)^{2(d+1)} \\
	&< 2^{2(d+1)\log(30)} \leq 2^{2(d+1)5} = 2^{|G|}
	\end{align*}
	which implies that $\Delta(\mathfrak{R}_{\fcal_{\ell_1}}) < 10(d+1)$.
\end{proof}

\begin{proof}[Proof of Lemma \ref{lem:l1_subspace_embedding}.]
	Consider any $\beta\in\REALS^d$. Let $D_wX=UR$ where $U$ is an orthonormal basis for the columnspace of $D_wX$. As in Lemma \ref{lem:one} we have for each index $i$
	\begin{align}
	\label{eqn:l1bound}
	|w_ix_i\beta| &= |U_i R\beta| \leq \twonorm{U_i} \twonorm{R\beta}=\twonorm{U_i} \twonorm{D_wX\beta} \leq \twonorm{U_i} \onenorm{D_wX\beta}
	\end{align}
	The sensitivity for the $\ell_1$ norm function of $x_i\beta$ is thus $$\sup_{\beta\in\REALS^d\setminus\{0\}}\frac{\!\!\!w_i|x_i\beta|}{\onenorm{D_wX\beta}}\leq \twonorm{U_i}.$$
	Note that our upper bounds on the sensitivities satisfy $s_i\geq \twonorm{U_i}$. Thus also $S=\sum\nolimits_{i=1}^n s_i \geq \sum\nolimits_{i=1}^n \twonorm{U_i}$ holds. In particular, these values are exceeded by more than a factor of $\mu+1$.
	Also, by Lemma \ref{lem:vc_ell1}, we have a bound of $O(d)$ on the VC dimension of the class of functions $\fcal_{\ell_1}$.
	Now, rescaling the error probability parameter $\delta$ that we put into Theorem \ref{thm:sensitivity} by a factor of $\frac 1 2$, and union bound over the two sets of functions $\fcal_{log}$, and $\fcal_{\ell_1}$, the sample in Theorem \ref{thm:basicalg} satisfies at the same time the claims of Theorem \ref{thm:basicalg} with parameter $\eps$ and of this lemma with parameter $\eps' \leq \eps/\sqrt{\mu + 1}$ by folding the additional factor of $\mu+1$ into $\eps$.
\end{proof}

\begin{proof}[Proof of Lemma \ref{lem:muerror}.]
	For brevity of presentation let $X'=D_wX$. First note that combining the choice of the parameter $\eps/\sqrt{\mu+1}$ with Lemma \ref{lem:l1_subspace_embedding} we have for all $\beta\in\REALS^d$ $$\left(1-{\eps'}\right)\onenorm{X'\beta}\leq\onenorm{TX'\beta}\leq \left(1+{\eps'}\right)\onenorm{X'\beta},$$ where $\eps'\leq \frac{\eps}{\mu+1}$. Note that since the weights are non-negative, sampling and reweighting does not change the sign of the entries. This implies for $\eta^+ = |\onenorm{(TX'\beta)^+}-\onenorm{(X'\beta)^+}|$ and $\eta^- = |\onenorm{(TX'\beta)^-}-\onenorm{(X'\beta)^-}|$ that $\max\{\eta^+,\,\eta^-\}\leq \eta^+ +\eta^- = |\onenorm{TX'\beta}-\onenorm{X'\beta}|\leq \eps'\onenorm{X'\beta}.$
	
	From this and $\onenorm{X'\beta}=\onenorm{(X'\beta)^+} + \onenorm{(X'\beta)^-}\leq (\mu+1) \min\{\onenorm{(X'\beta)^+},\onenorm{(X'\beta)^-}\}$ it follows for any $\beta\in\REALS^d$
	\begin{align*}
	\frac{\onenorm{(TX'\beta)^+}}{\onenorm{(TX'\beta)^-}} &\leq \frac{\onenorm{(X'\beta)^+}+\eps'\onenorm{X'\beta} }{\onenorm{(X'\beta)^-}-\eps'\onenorm{X'\beta}} \leq \frac{\onenorm{(X'\beta)^+}+\eps'(\mu+1)\onenorm{(X'\beta)^+} }{\onenorm{(X'\beta)^-}-\eps'(\mu+1)\onenorm{(X'\beta)^-}}\\
	&\leq \frac{\onenorm{(X'\beta)^+}(1+\eps) }{\onenorm{(X'\beta)^-}(1-\eps)} \leq \mu\, \frac{1+\eps}{1-\eps} \leq (1+4\eps) \mu\,.
	\end{align*}
	The claim follows by folding the constant $\frac{1}{4}$ into $\eps$.
\end{proof}

\begin{proof}[Proof of Theorem \ref{thm:recalg}.]
	Recall, due to Lemma \ref{lem:muerror}, the $\mu'$-complexity at the $i$-th recursion level is upper bounded by $\mu(1+\eps)^i$. We thus apply Theorem \ref{thm:basicalg} recursively $l=\log\log n$ times with parameter $\eps_i = \frac{\eps}{2l\sqrt{\mu+1}(1+\eps)^{i}}$ for $i\in\{0\ldots l-1\}$. 
	First we bound the approximation ratio, which is the product of the single stages. We have
	\begin{align*}
	\prod\limits_{i=0}^{l-1}(1+\eps_i)
	&\leq \prod\limits_{i=0}^{l-1}\left(1+\frac{\eps}{(1+\eps)^i2l\sqrt{\mu}}\right) \leq \left(1+\frac{\eps}{2l\sqrt{\mu}}\right)^l \leq \exp\left( \frac{\eps}{2\sqrt{\mu}} \right) \leq 1+\frac{\eps}{\sqrt{\mu}}.
	\end{align*}
	Also 
	\begin{align*}
	\prod\limits_{i=0}^{l-1}(1-\eps_i)
	&\geq \prod\limits_{i=0}^{l-1}\left(1-\frac{\eps}{(1+\eps)^i2l\sqrt{\mu}}\right) \geq \prod\limits_{i=0}^{l-1}\left(1-\frac{\eps}{2l\sqrt{\mu}}\right)
	\geq 1-\sum\limits_{i=0}^{l-1}\frac{\eps}{2l\sqrt{\mu}} \geq 1-\frac{\eps}{2\sqrt{\mu}}.
	\end{align*}
	
	Initially all weights are equal to one. So in the first application of Theorem \ref{thm:basicalg} we have $\omega = 1$. This value might grow as the weights are reassigned. However, from Inequality (\ref{eqn:sensitivitypowers}) and the discussion below it follows, that the value of $\omega$ can grow only by a factor of $2n$ in each recursive iteration. So it remains bounded by $\omega \leq (2n)^{\log\log n}$ in all levels of our recursion. Its contribution to the lower order terms given in Theorem \ref{thm:basicalg} is thus bounded by $O\left(\log ((2n)^{1+\log\log n})) \right)\subseteq O\left(\log n \log \log n \right).$
	
	The size of the data set at recursion level $i+1$ satisfies 
	\begin{align*}
	n_{i+1} &\leq \sqrt{n_{i}}\cdot \frac{Cl^2(1+\eps)^{2i}\mu^2}{\eps^2} d^{3/2}\log((1+\eps)^{i}\mu n d) \log n \log \log n\\ &\leq \sqrt{n_{i}}\cdot \frac{Cl^2 4^{i}\mu^2}{\eps^2} d^{3/2}\log(2^i \mu n d) \log n \log \log n
	\end{align*}
	for some constant $C>1$. Solving the recursion until we reach $n_0=n$ we get the following bound on $n_l$. We use that for our choice $l=\log\log n$ we have $2^l = \log n$ and $n^{2^{-l}}=2^{\frac{\log n}{2^l}} = 2$.
	\begin{align*}
	n_l&\leq n^{2^{-l}}\prod\limits_{i=0}^l \left( C \cdot \frac{l^2 4^i \mu^2}{\eps^2} d^{3/2}\log(2^i\mu n d) \log n \log \log n \right)^{\frac{1}{2^{i}}}\\
	&\leq 2\prod\limits_{i=0}^l 4^{\frac{i}{2^{i}}} \prod\limits_{i=0}^l
	\left( C\cdot\frac{l^2 \mu^2}{\eps^2} d^{3/2}\log(2^l\mu n d) \log n \log \log n \right)^{\frac{1}{2^{i}}}\\
	&\leq 2\prod\limits_{i=0}^l 4^{\frac{i}{2^{i}}} \prod\limits_{i=0}^l
	\left( 2C\cdot \frac{l^2 \mu^2}{\eps^2} d^{3/2}\log(\mu n d) \log n \log \log n \right)^{\frac{1}{2^{i}}}\\
	&\leq 2\cdot 4^{\sum\limits_{i=0}^l \frac{i}{2^{i}}} \left( 2C\cdot \frac{l^2 \mu^2}{\eps^2} d^{3/2}\log(\mu n d) \log n \log \log n \right)^{\sum\limits_{i=0}^l \frac{1}{2^{i}}}\\
	&\leq 2\cdot 4^{2} \left( 2C\cdot \frac{l^2\mu^2}{\eps^2} d^{3/2}\log(\mu n d) \log n \log \log n \right)^{2}\\
	&\leq 2\cdot 16\cdot 4C^2 \cdot \frac{l^4 \mu^4}{\eps^4} d^{3}\log^2(\mu n d) \log^2 n \, (\log \log n)^2
	\end{align*}
	We conclude that for some constant $C'>C$
	\begin{align*}
	n_l &\leq C'\cdot\frac{l^4\mu^4}{\eps^4} d^{3}\log^2(\mu n d) \log^2 n  \, (\log \log n)^2 \leq C'\cdot\frac{\mu^4}{\eps^4} d^{3}\log^2(\mu n d)\log^2 n \, (\log \log n)^6.
	\end{align*}
	To reduce this even further, note that in the final iteration we do not need to preserve the $\mu$-complexity. We can thus apply Theorem \ref{thm:basicalg} with the original approximation parameter $\eps$ to obtain a coreset as claimed of size
	\begin{align*}
	k &\in O\left( \sqrt{n_l} \cdot \frac{\mu}{\eps^2} d^{3/2}\log(\mu n d)\log(\omega n) \right)\\
	&\subseteq O\left( \frac{\mu^2}{\eps^2} d^{3/2}\log(\mu n d)\log n \, (\log \log n)^3 \cdot \frac{\mu}{\eps^2} d^{3/2}\log(\mu n d)\log n \, (\log\log n) \right)\\
	&\subseteq O\left( \frac{\mu^3}{\eps^4} d^{3}\log^2(\mu n d)\log^2 n \, (\log \log n)^4 \right) .
	\end{align*}
	
	It remains to bound the failure probability. Note that we use a $\log n$ factor in the sampling sizes at all stages rather than $\log n_i$. The failure probability at each stage is thus bounded by $\frac{1}{n^{c'}}$ for $c' = c+1 > 2$ by adjusting constants. We can thus take a union bound over the stages to get an error probability of at most $$l\cdot \frac{1}{n^{c'}}=\frac{\log\log n }{n^{c'}}\leq\frac{1}{n^{c'-1}} \leq \frac{1}{n^{c}}.$$

	Now recall from Theorem \ref{thm:basicalg} the two pass streaming algorithm whose running time was dominated by $O(\nnz(X)\log n_i+ \poly(d)\log n_i)$. We can thus bound the running time of the recursive algorithm for sufficiently large $C>1$ by
	\begin{align*}
		C (\nnz(X)+\poly(d)) \sum\nolimits_{i=0}^{l-1} \log n_i &\leq C (\nnz(X)+\poly(d)) \log n \log\log n \\
		&\in O((\nnz(X)+\poly(d))\log n\log\log n).
	\end{align*}
	Regarding the number of passes, note that for any $\eta>0$, after $\log(\frac{1}{\eta})$ recursion steps, the leading term in the size of the coreset is as low as $n^{2^{-\log\frac{1}{\eta}}}=n^\eta$, after which we may arguably assume, that the coreset fits into memory. The algorithm thus takes $2\log(\frac{1}{\eta})$ streaming passes over the data before it turns to an internal memory algorithm.
\end{proof}

\clearpage
\section{Material for the experimental section}
\label{app:experiments}
\begin{table}[ht!]
	\caption{Absolute values of the negative log-likelihood $\mathcal L(\beta_{opt})$ at the optimal value $\beta_{opt}$ and mean running time $t_{opt}$ in seconds from the optimization task on the full data sets.}
	\label{tbl:abs_vals}
	\vskip 0.15in
	\begin{center}
		\begin{small}
			\begin{sc}
				\begin{tabular}{lrr}
					\toprule
					Data set &{ $\mathcal L(\beta_{opt})$ }& $t_{opt}$ \\
					\midrule
					Webb Spam   &    69,534.49 & 1,051.72 \\
					Covertype   &   270,585.34 &   218.22 \\
					KDD Cup '99 &   301,023.24 &   136.06 \\
					\bottomrule
				\end{tabular}
			\end{sc}
		\end{small}
	\end{center}
	\vskip -0.1in
\end{table}

\begin{figure*}[ht!]
	\vskip 0.2in
	\begin{center}
		\begin{sc}
			\begin{tabular}{ccc}
				{\small\hspace{.5cm}Webb Spam}&{\small\hspace{.5cm}Covertype}&{\small\hspace{.5cm}KDD Cup '99} \\
				\includegraphics[width=0.30\linewidth]{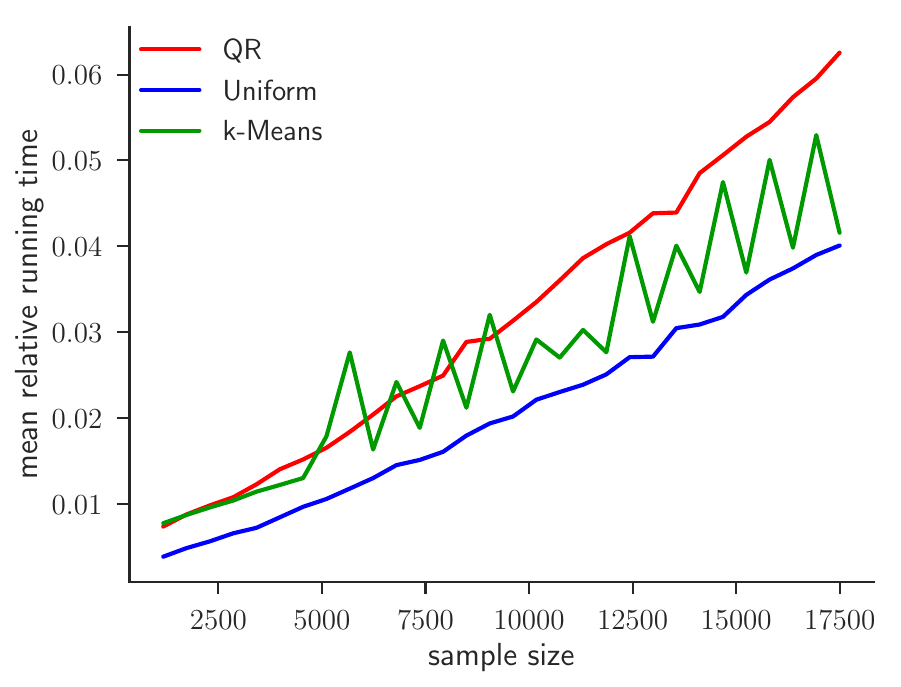}&
				\includegraphics[width=0.30\linewidth]{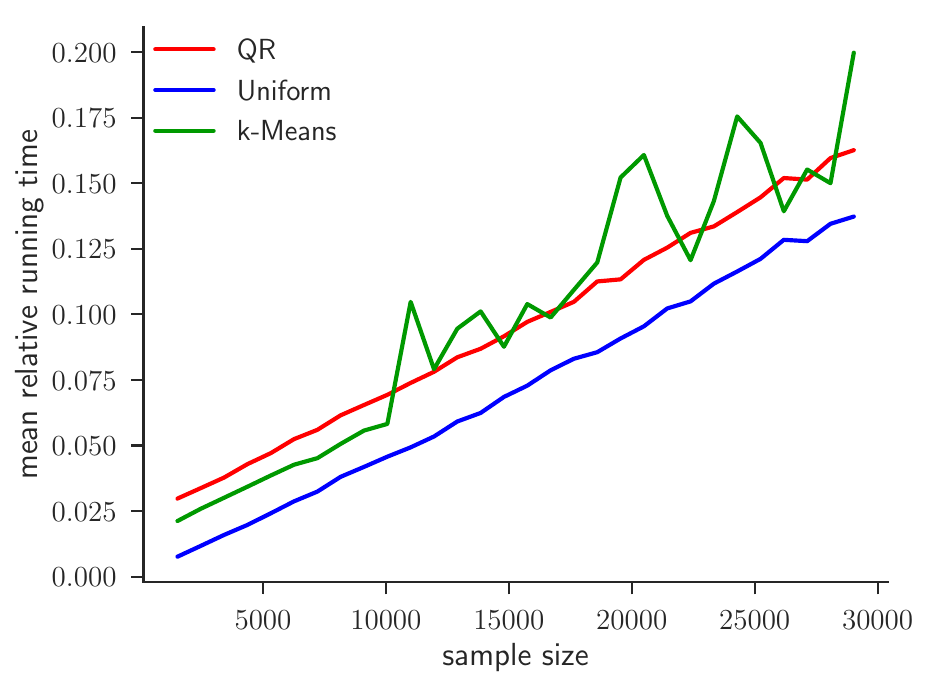}&
				\includegraphics[width=0.30\linewidth]{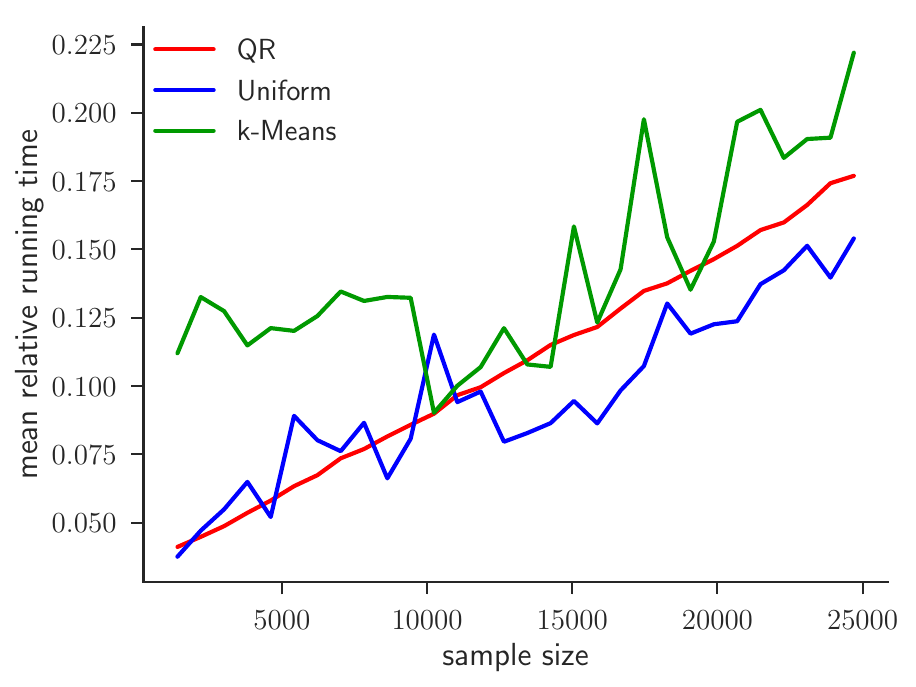} \\
				\includegraphics[width=0.30\linewidth]{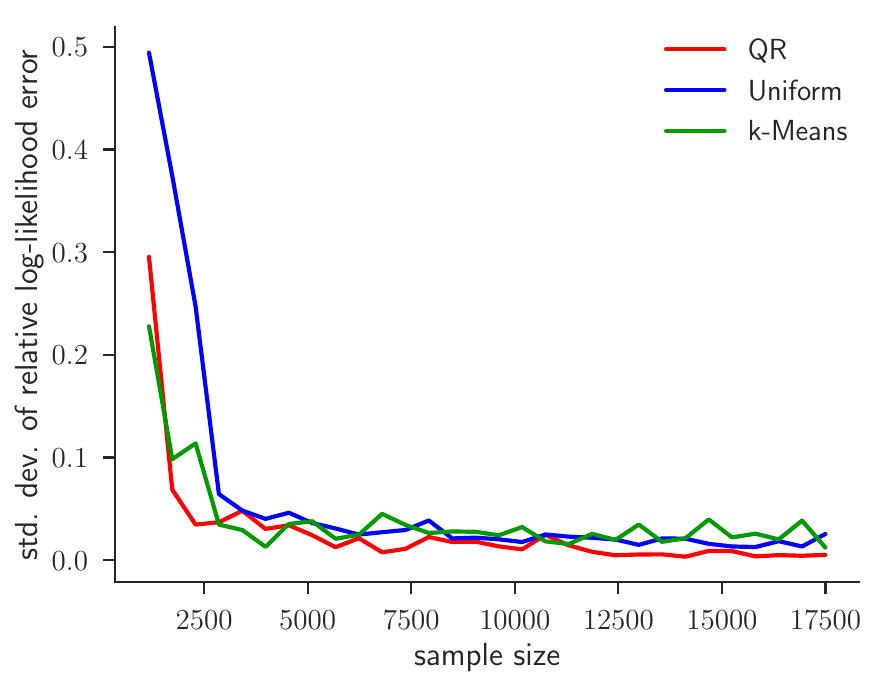}&
				\includegraphics[width=0.30\linewidth]{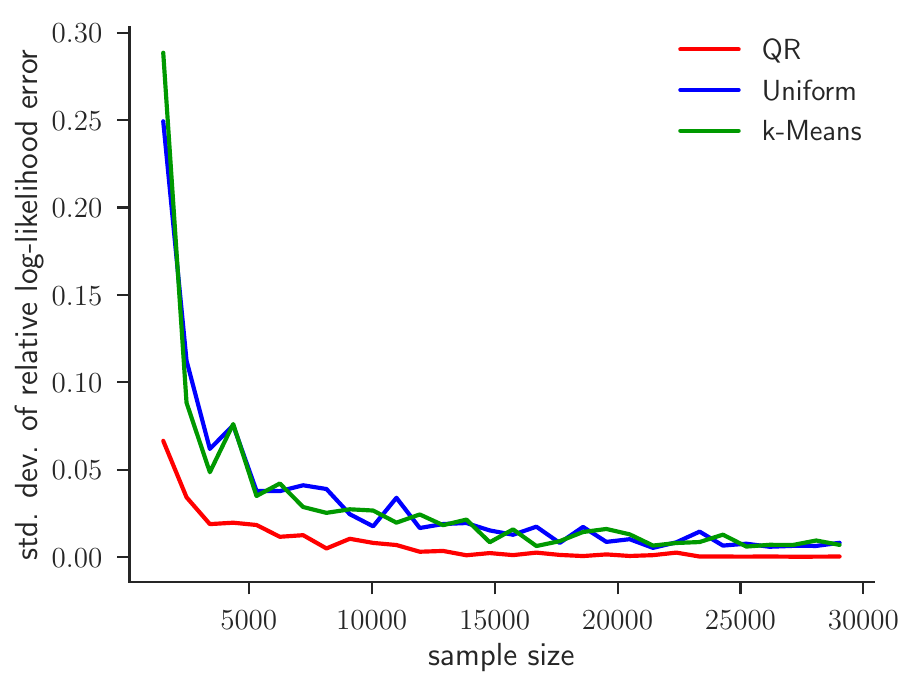}&
				\includegraphics[width=0.30\linewidth]{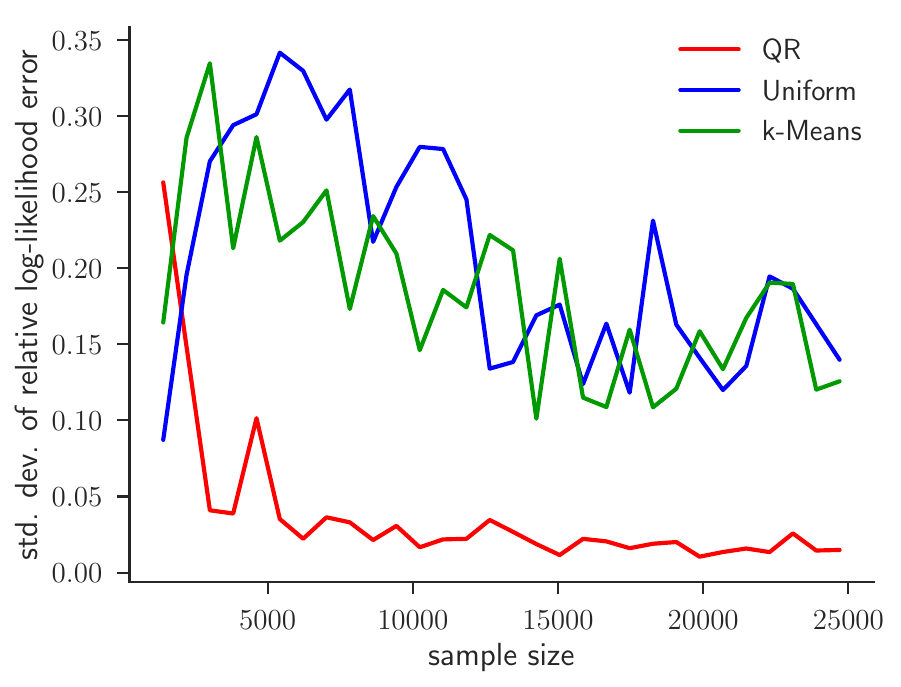} \\ 
				\includegraphics[width=0.30\linewidth]{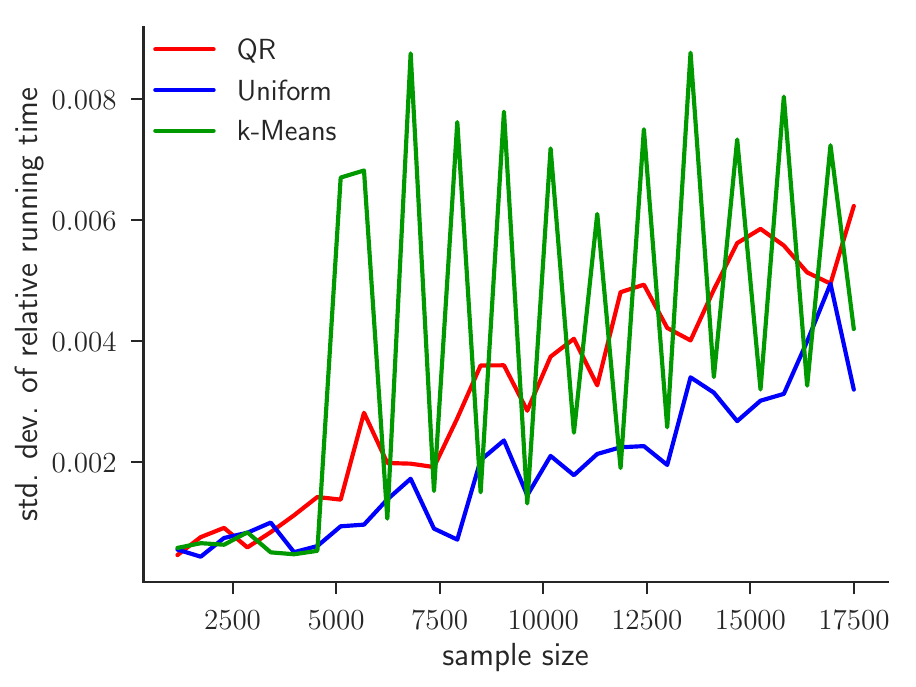}&
				\includegraphics[width=0.30\linewidth]{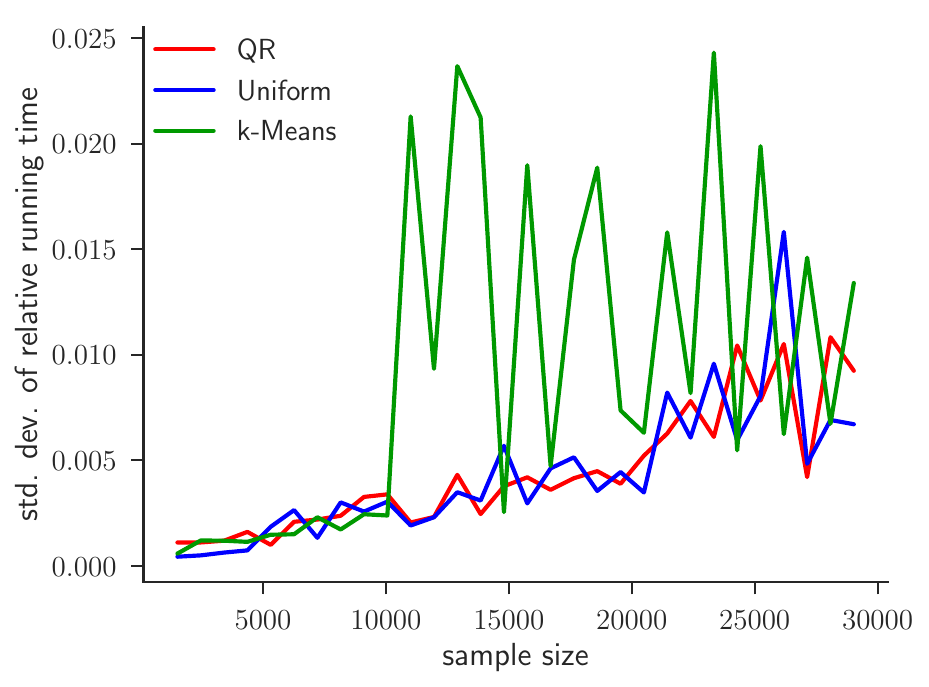}&
				\includegraphics[width=0.30\linewidth]{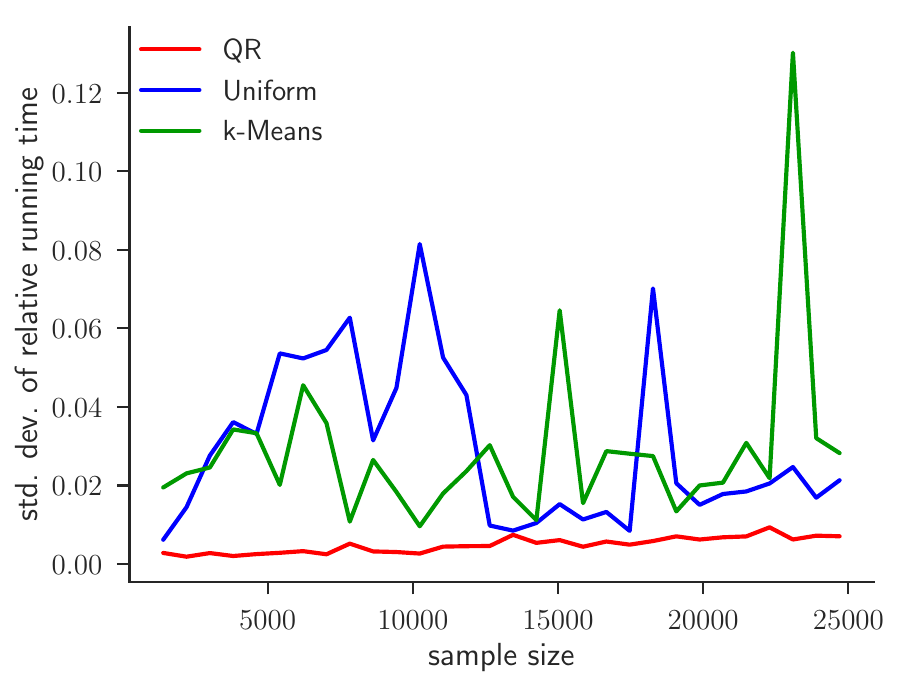} \\
			\end{tabular}
		\end{sc}
		\caption{Each column shows the results for one data set comprising thirty different coreset sizes (depending on the individual size of the data sets). The plotted values are means and standard deviations taken over twenty independent repetitions of each experiment. The plots show the mean relative running times (upper row), the standard deviations of the relative log-likelihood errors (middle row) and standard deviations of the relative running times (lower row) of the three subsampling distributions, uniform sampling (blue), our QR derived distribution (red), and the $k$-means based distribution (green). All values are relative to the corresponding running times respectively optimal log-likelihood values of the optimization task on the full data set, see Table \ref{tbl:abs_vals} (lower is better).}
		\label{fig:supp:variances}
	\end{center}
	\vskip -0.2in
\end{figure*}

\clearpage
\section{Discussion of uniform sampling}
\label{dis:uniform}
As we have discussed in the lower bounds section \ref{sec:LBs}, uniform sampling cannot help to build coresets of sublinear size for worst case instances. Actually this also holds for other techniques for solving logistic regression that rely on uniform subsampling, such as stochastic gradient descent (SGD).

We support this claim via a little experiment and some theoretical discussion on the following data set $X$ of size $m=2n+2$ in one dimension (plus intercept): The $-1$ class consists of one point at $-n$ and $n$ points at $1$, while class $+1$ consists of one point at $+n$ and $n$ points at $-1$. By symmetry of $\ell_1$-norms, it is straightforward to check that the data is $\mu$-complex for $\mu=1$, and the optimal solution is $\hat\beta=0$, which corresponds to $f(X\hat\beta)=\sum\nolimits_{i=1}^{m} \ln(1+\exp(0))=m\ln(2) < m$. Our algorithms will thus find a coreset of sublinear size such that the optimal solution has a value of at most $(1+\eps) m$ with high probability.

A uniform sample of sublinear size misses the two points at $-n$ and $n$, since the probability to sample one of these is $\frac{1}{n+1}$. However, finding these points is crucial, since otherwise the remaining data is separable, which leads to a large $\beta$. Adding penalization is not a remedy. Figure \ref{fig:unif} shows the results of running $1\,000$ independent repetitions of \texttt{sklearn.linear\_model.SGDClassifier} in Python for logistic regression, with $\ell^2_2$-penalty enabled, on the data set with $n=50\,000$. The boxplots show the resulting coefficients for the intercept ${\beta}_0$ and for the single dimension ${\beta}_1$. One might argue that the intercept term is close to $\hat\beta_0=0$, but for $\beta_1$, half of the values lie above the median of $107.06$ (red line) and still a quarter lies even above the upper quartile of $492.35$ (upper boundary of the box).

Note that assuming $\beta_0 = 0$ and $\beta_1\gg (1+\eps)$, we have $f(X\beta)>2n\beta_1\gg (1+\eps)m$, since by construction
\begin{align*}
	f(X\beta) &=\sum\nolimits_{i=1}^{m} \ln(1+\exp(-y_i(\beta_0+x_i\beta_1))) \\
	&\geq 2 n\cdot \ln(1+\exp(\beta_0+\beta_1)) \\
	&\geq 2 n\beta_1.
\end{align*}
This implies that the approximation ratio is $\frac{f(X\beta)}{f(X\hat\beta)} \geq \frac{2n\beta_1}{m}=\frac{2n\beta_1}{2n+2}\overset{n\rightarrow \infty}{\longrightarrow} \beta_1$, which turned out very large in the experiment above, cf. Figure \ref{fig:unif}.

\begin{figure}[ht!]
	\centering
	\includegraphics[width=.6\textwidth]{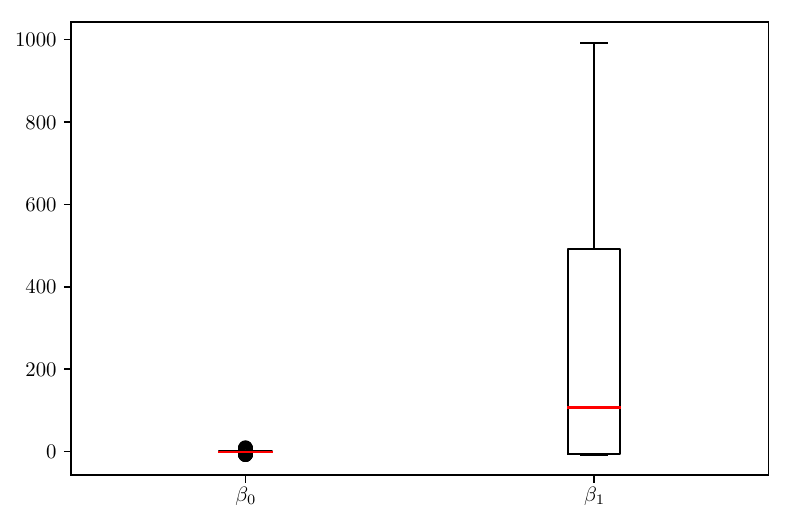}
	\caption{Boxplots of the solutions $\beta=(\beta_0,\beta_1)$ for logistic regression found by SGD in $1\,000$ independent runs on the considered data set $X$. The optimal solution is $\hat\beta=(0,0)$. It can be seen that while the intercept term $\beta_0$ is reasonably close to $0$, the majority of runs result in considerably large values of $\beta_1$, which leads to a bad approximation ratio.}
	\label{fig:unif}
\end{figure}
\end{document}